\documentclass[10pt,doublecolumn]{IEEEtran}
\usepackage[font=small,labelsep=space]{caption}
\usepackage{verbatim}
\usepackage{graphicx}
\usepackage{subfig}
\usepackage{multirow}
\usepackage{amssymb}
\usepackage{amsmath}
\usepackage{amsthm}
\usepackage{amsfonts}
\usepackage{paralist}
\usepackage{url}
\usepackage{epstopdf}
\usepackage{float}
\usepackage[justification=centering]{caption}
\usepackage{color}
\usepackage{paralist}
\usepackage[noend]{algpseudocode}
\usepackage{algorithmicx}
\usepackage{algorithm}
\usepackage{mathrsfs}
\usepackage{xcolor}
\let\itemize\compactitem
\let\enditemize\endcompactitem
\let\enumerate\compactenum
\let\endenumerate\endcompactenum

\begin{document}
\title{QoE-Driven UAV-Enabled Pseudo-Analog Wireless Video Broadcast: A Joint Optimization of Power and Trajectory}
\author{Xiao-Wei~Tang,~\IEEEmembership{Student Member,~IEEE}, Xin-Lin~Huang*,~\IEEEmembership{Senior Member,~IEEE}, and Fei Hu,~\IEEEmembership{Member,~IEEE}
\thanks{Xiao-Wei Tang (email: {\tt xwtang@tongji.edu.cn}) is with the Department of Control Science and Engineering, Tongji University, Shanghai 201804, China.}
\thanks{Xin-Lin Huang (email: {\tt xlhuang@tongji.edu.cn}) is with the Department of Information and Communication Engineering, Tongji University, Shanghai 201804, China (Corresponding Author).}
\thanks{Fei Hu (e-mail: {\tt fei@eng.ua.edu}) is with the Department of Electrical and Computer Engineering, University of Alabama, Tuscaloosa, AL 35487 USA.}
}
\markboth{}{Tang \MakeLowercase{\textit{et al.}}: QoE-Driven UAV-Enabled Pseudo-Analog Wireless Video Broadcast: A Joint Optimization of Power and Trajectory \ldots}
\maketitle
\maketitle
\begin{abstract}
The explosive demands for high quality mobile video services have caused heavy overload to the existing cellular networks. Although the small cell has been proposed to alleviate such a problem, the network operators may not be interested in deploying numerous base stations (BSs) due to expensive infrastructure construction and maintenance. The unmanned aerial vehicles (UAVs) can provide the low-cost and quick deployment, which can support high-quality line-of-sight communications and have become promising mobile BSs. In this paper, we propose a quality-of-experience (QoE)-driven UAV-enabled pseudo-analog wireless video broadcast scheme, which provides mobile video broadcast services for ground users (GUs). Due to limited energy available in UAV, the aim of the proposed scheme is to maximize the minimum peak signal-to-noise ratio (PSNR) of GUs’ video reconstruction quality by jointly optimizing the transmission power allocation strategy and the UAV trajectory. Firstly, the reconstructed video quality at GUs is defined under the constraints of the UAV’s total energy and motion mechanism, and the proposed scheme is formulated as a complex non-convex optimization problem. Then, the optimization problem is simplified to obtain a tractable suboptimal solution with the help of the block coordinate descent model and the successive convex approximation model. Finally, the experimental results are presented to show the effectiveness of the proposed scheme. Specifically, the proposed scheme can achieve over 1.6dB PSNR gains in terms of GUs’ minimum PSNR, compared with the state-of-the-art schemes, e.g., DVB, SoftCast, and SharpCast.
\end{abstract}

\IEEEpeerreviewmaketitle
\section{Introduction}
\IEEEPARstart {A}{ccording} to Cisco's latest report [1], mobile video traffic accounted for $59\%$ of global mobile data traffic in 2017 and will reach $79\%$ by 2022. In particular, $80\%$ of the mobile video traffic belongs to hotspot contents (e.g., the American Super Bowl, the World Cup, etc.) which are usually requested by many users simultaneously. How to efficiently deliver these hotspot videos to ground users (GUs) becomes a non-trivial problem. Point-to-point (P2P) communication is first excluded due to its high power consumption and low bandwidth utilization. Broadcast communication can efficiently save the bandwidth and the transmission power, which makes it seem like a better choice compared with P2P [2]. In the conventional digital video broadcast (DVB) system, the base station (BS) transmits the video at the rate which matches the throughput of the GU with the worst channel quality to ensure that most users can successfully decode the video. However, the cliff effect may occur in DVB when GUs' channel qualities vary a lot [3]. The pseudo-analog video transmission (PAVT) techniques have been proposed to effectively alleviate the cliff effect [4-8]. Specifically, PAVT adopts the joint source-channel coding instead of separate source and channel coding [9,10]. Since video signals are processed linearly in PAVT, the video demodulation quality is approximately proportional to the corresponding channel quality. Therefore, PAVT can provide continuous scalable video quality, which makes it suitable for video broadcast [11,12].

Many novel schemes on PAVT have been proposed in recent years. In [7], Katabi et al. proposed a cross-layer design for wireless video broadcast named SoftCast which was the first work on PAVT. In [8], He et al. proposed a structure-preserving video delivery system named SharpCast to improve both the objective and subjective visual quality. In [9], a PAVT scheme called D-Cast was proposed where the correlation among videos could be fully utilized. In [10], a data-assisted cloud radio access PAVT network named DAC-RAN was proposed, which separated the control and data planes in the conventional digital transmission infrastructure, and integrated a new data plane into the virtual BS. In [11], Huang et al. proposed a knowledge-enhanced wireless video transmission system called KMV-Cast which could exploit the hierarchical Bayesian model to integrate the correlated information into the video reconstruction.

The above PAVT schemes provide new feasible solutions to video broadcast and users can enjoy the video qualities matching their own channel qualities. However, PAVT schemes can not solve this problem, that is, cell-edge users far away from BSs suffer from unsatisfactory quality of experience (QoE) due to poor channel qualities [13,14]. Although the cellular offloading technique has been proposed to alleviate such a problem, its high cost of deploying new BSs makes it unsuitable for the scenarios with temporary traffic (e.g., an on-demand concert) [15,16]. Furthermore, the cellular offloading technique is usually applied in the conventional network architectures with fixed infrastructure such as ground BSs, access points and relays, which makes it difficult to serve GUs in an energy-efficient manner.

Nowadays, unmanned aerial vehicles (UAVs), which can provide Internet access from the sky, have been considered as a promising solution to improve the communication quality of cell-edge GUs. Particularly, the technology advances in aviation and artificial intelligence have enhanced the functions of UAVs and made them smaller, lighter, and smarter [17-19]. Therefore, UAVs have been widely used to implement complex tasks such as post-disaster reconnaissance, agriculture precision, cargo transportation, etc., due to their advantages including rapid deployment, high flexibility, controllable mobility, and  easily available line-of-sight (LoS) channels [20-24]. Specially, LoS channels can make the communication process suffer less path loss, shadowing and fading [25-27]. In addition, the coverage area of the UAV can be easily adjusted by changing the UAV's height, transmission power, and antenna orientations.

To the best of our knowledge, the existing studies on UAV-enabled video transmission are generally limited to the conventional digital systems, e.g., P2P and DVB. Therefore, these schemes inevitably suffer from some problems, e.g., lack of scalability/reliability, caused by the inherent defects of the digital systems. For example, in the UAV-enabled P2P video transmission scenario, the UAV usually stores multiple versions of the same video content encoded at different coding rates. The GU can only request the UAV for the video content which matches its instantaneous end-to-end throughput [28,29] or playback buffer status [30]. Switching between different video rates frequently inevitably causes visual quality fluctuations that affect the QoE [31,32]. Moreover, the UAV's high mobility often results in time-varying channel quality and network topology. The video demodulation quality reacts strongly to the link degradation when streaming videos over such unreliable channels, and a single packet loss may lead to video freezing for several seconds [33-37]. In addition, existing work simply assumes that the GUs' QoE only depends on the video transmission rate, and do not measure the video demodulation quality from the perspective of objective image evaluation metric.

Motivated by the above limitations of existing studies, we propose a QoE-driven UAV-enabled pseudo-analog wireless video broadcast (QUPWV-Cast) system, in which the UAV is dispatched as a mobile BS to provide video broadcast services for a group of GUs. The goal is to maximize the minimum peak signal-to-noise ratio (PSNR) of GUs' video demodulation quality by jointly optimizing the transmission power allocation strategy and the UAV trajectory. \\
\textbf{Contributions}: The main contributions of this paper are two-fold:
\begin{enumerate}
\item This is the first work to integrate the UAV technique into the PAVT system. UAV's advantages including high mobility, promising LoS channel, and fast deployment are fully utilized to enhance the QoE of cell-edge GUs, which provides a feasible solution to video broadcast from the sky.
\item The proposed QUPWV-Cast system is modelled as a non-convex optimization problem to maximize the minimum PSNR of GUs by jointly optimizing the transmission power allocation strategy and the UAV trajectory. Specifically, an effective optimization algorithm based on the block coordinate descent (BCD) and successive convex approximation (SCA) techniques is proposed to obtain the sub-optimal solution.
\end{enumerate}

The reminder of the paper is organized as follows. Section II introduces related work and fundamental knowledge on the proposed QUPWV-Cast. In Section III, the details of the system model are presented. In Section IV, an effective BCD and SCA-based algorithm is proposed to maximize the minimum PSNR of GUs. In Section V, the simulation results are provided to show the effectiveness of the proposed system. In Section VI, we conclude this paper.

{\bf{Notations}}: in this paper, italics represent scalars, bold lowercases represent vectors, and bold uppercase represents sets. ${{\mathbb R}^{M \times 1}}$ represents a $M$-dimensional real-value vector. For vector ${\bf{a}}$, $||{\bf{a}}||$ represents its Euclidean norm, and ${{\bf{a}}^T}$ represents its transpose. $E(\cdot)$ represents the expectation of a random variable.

\section{Related Work and Fundamental Knowledge}
In this section, we will first introduce related work on the UAV-enabled video transmission. Then, we will give fundamental knowledge on the PAVT scheme and UAV-GU channel model.

\subsection{UAV-Enabled Video Transmission}
UAVs can capture videos via their equipped sensors, and then deliver them to the GUs after compression and encoding. Therefore, UAVs play an important role in many video streaming applications such as live streaming, virtual reality/augmented reality, etc. These applications usually have high QoE requirements, such as low packet loss ratio (PLR) and ultra-high resolution [38-41]. Specifically, there are two key issues in the design and implementation of UAV-enabled wireless video transmission: 1) The UAV trajectory needs to be properly designed [42], so that the UAV can approach GUs as closely as possible to obtain better channel quality and reduce energy consumption; 2) The limited resource, e.g., energy, should be properly allocated during the UAV's flight [43,44], in order to maintain the QoE and reduce the communication outage probability.

Many studies have contributed to the development of UAV-enabled wireless video transmission by solving the above two key issues. In [45], Wu et al. investigated a UAV-enabled orthogonal frequency division multiple access (OFDMA) scheme, in which the UAV was adopted as a BS. A minimum-rate ratio (MRR) was defined for each GU to represent the minimum required instantaneous rate to maintain the average throughput. The goal was to maximize the minimum average throughput of all GUs by jointly optimizing the UAV trajectory and OFDMA resource allocation under the given constraints on GUs' MRRs. In [46], Zeng et al. studied a UAV-enabled multicasting system, in which a UAV transferred a common file to a set of GUs. The goal was to minimize the mission completion time by designing the UAV trajectory, while ensuring that each GU could successfully recover the file with a desired probability. In [47], Ludovico et al. improved the delay performance for video streaming applications in congested cellular macro-cells using a mobile micro-cell installed on a UAV. The mobile micro-cell was used to offload GUs from a congested macro-cell to optimize the bandwidth usage. In [48], He et al. designed an intelligent and distributed allocation mechanism to improve GUs' QoE. Each UAV in a cluster could independently adjust and select its video encoding rate to achieve flexible uplink allocation. They built a potential game model to maximize the GUs' QoE through a low-complexity distributed self-learning algorithm. In [49], Zhan et al. extended the UAV applications to adaptive streaming services over fading channels. The objective was to maximize the overall utility while guaranteeing the fairness among multiple GUs under the constraints of UAV energy budget and rate outage probability. In [50], Wu et al. considered UAV-enabled two-user broadcast channel, where a UAV was deployed to send independent information to two GUs at different fixed locations. It aimed to characterize the capacity region over a given UAV flight duration, by jointly optimizing the UAV's trajectory and transmission power/rate allocations over time, subject to the UAV's maximum speed/power constraints. In [51], Colonnese et al. investigated the benefits of flexible resource allocation when performing adaptive video streaming across cellular systems. To guarantee the video smoothness in the presence of fluctuations of the channel capacity, the authors considered a proxy video manager and resource controller located at the cellular BS. A cross-layer bandwidth allocation scheme was proposed to minimize the transmission delay, considering the channel quality, video quality requirements and coding rate fluctuations.

In summary, the above studies mainly focus on the UAV-enabled digital video transmission. They cannot avoid inherent defects of the digital systems, such as the low bandwidth usage in P2P and the cliff effect in DVB. In order to provide an effective solution to UAV-enabled broadcast scenario, we will introduce the fundamental knowledge of PAVT scheme in the next part.

\subsection{PAVT Scheme}
Conventional DVB systems with JPEG 2000/H.264 divide a video into groups of pictures (GOPs) and adopt predictive coding. Then, intra- and inter-frame correlations of the video enable high compression efficiency. The DVB system can select a suitable channel modulation/coding scheme (MCS) to overcome the channel interference based on the channel conditions. However, due to the fluctuations of the channel quality, the selected MCS may not guarantee a predetermined PLR. The cliff effect will occur under deep channel fading. Especially in the broadcast scenario, the receivers have different channel qualities, and the transmitter should select a MCS according to the worst channel quality to guarantee the correct demodulation of all receivers.

The PAVT scheme can be used to improve the effectiveness of video broadcast, in which the transmitter does not need to know the receivers' channel quality since the video demodulation quality is proportional to its corresponding channel quality. Specifically, a video is first divided into multiple GOPs. Then, three-dimensional discrete cosine transform (3D-DCT) is performed for each GOP to remove the spatial and temporal redundancy. Next, the transformed DCT coefficients in each frame are divided into blocks with a uniform size and the transmitter allocates different transmission power levels for each coefficient block according to its variance. After that, Hadamard transform is performed for each block to reduce the peak-to-average power ratio (PAPR). Finally, the transmitter sends those transformed coefficients in high-density modulation mode. At each receiver, it performs a series of operations in order, including the signal demodulation, inverse Hadamard transform, minimum mean squared error estimation, inverse 3D-DCT transform, and finally reconstructs the entire video.

\subsection{UAV-GU Channel Model}
Consider the case with one UAV and $N$ GUs (denoted as ${u_n, n \in {\cal N} = \{ 1,...,N\}}$), where the coordinate of $u_n$ is known to the UAV and can be denoted as ${{{{\bf{w}}}}_n} = {\left[ {{x_n},{y_n},0} \right]^T},\;\forall n \in {\cal N}$. For ease of analysis, the flight duration of the UAV is divided into $K$ time slots, and the length of each time slot $k \in {\cal K} = \{ 1,...,{K}\}$ is ${\Delta}$. Note that ${\Delta}$ can be set small enough so that the distance between the UAV and $u_n$ can be regarded as a constant in each time slot. Assume that the starting and ending points of UAV's flight are predetermined, which can be denoted as ${{\bf{w}}_0}$ and ${{\bf{w}}_F}$, and the UAV flies at the altitude $H$. Consequently, the UAV trajectory can be approximately represented by the set ${\bf{Q}}= \{{\bf{q}}[k]={\left[ {x[k],y[k], H} \right]^T},\forall k \in {\cal K}\}$, where ${\bf{q}}[k]$ represents the coordinate of the UAV in time slot $k$. The distance between the UAV and $u_n$ in time slot $k$ can be denoted as
\begin{equation}\label{E1}
{d_k}(n) = \left\| {{\bf{q}}[k] - {{\bf{w}}_n}} \right\|.
\end{equation}

Then, the average channel power gain ${\beta _k}(n)$ can be modelled as
\begin{equation}\label{E2}
{\beta _k}(n) = {\beta _0}d_k^{ - \alpha }(n) = \frac{{{\beta _0}}}{{{{(||{\bf{q}}[k] - {{\bf{w}}_n}|{|^2})}^{\alpha /2}}}},
\end{equation}
where ${\beta_0}$ is the average channel power gain at a reference distance of $d_0 = 1$m, and $\alpha$ is the path loss exponent that usually has a value between 2 and 6. Thus, the instantaneous channel gain from the UAV to $u_n$ in time slot $k$ can be denoted as
\begin{equation}\label{E3}
{h_k}(n) = {g_k}(n)\sqrt {{\beta _k}(n)},
\end{equation}
where ${g_k}(n)$ indicates the shadowing and small-scale fading component between the UAV and $u_n$ in time slot $k$, and $E[|{g_k}(n){|^2}] = 1$. By considering different channel models such as probabilistic LoS model and Rician fading model [52], we may rewrite ${g_k}(n)$ as different functions. For simplicity, we assume that the communication links from the UAV to GUs are dominated by LoS channels, and the Doppler effect caused by the UAV's mobility can be perfectly compensated at the receiver. Therefore, to simplify the analysis, we assume that ${g_k}(n) = 1$ and $\alpha  = 2$ in this paper.

\section{System model}
The scenario of the considered QUPWV-Cast is shown in Fig. 1, where a fixed-wing UAV\footnote{The fixed-wing UAV consumes less propulsion energy than the rotor UAV, and has a longer communication range.} is dispatched as a mobile BS to transmit the stored video to $N$ GUs through OFDM system. The flight duration of the UAV is divided into $K$ time slots. For simplicity, we take the time of transmitting a DCT coefficient block through the OFDM system as the length of each time slot $k$. Assume that the stored video has been processed as in SoftCast [7]. The goal of the proposed system is to maximize the minimum video demodulation quality of GUs by jointly optimizing the transmission power allocation strategy and the UAV trajectory.

In the following part, we will first introduce two basic models including distortion estimation model and energy consumption model. Then, we will define the PSNR-based objective function according to the estimated distortion. Finally, we will state the optimization problem under the constraints of UAV's total energy and motion mechanism.
\begin{figure}[htbp!]
\centering
\includegraphics[width=0.4\textwidth]{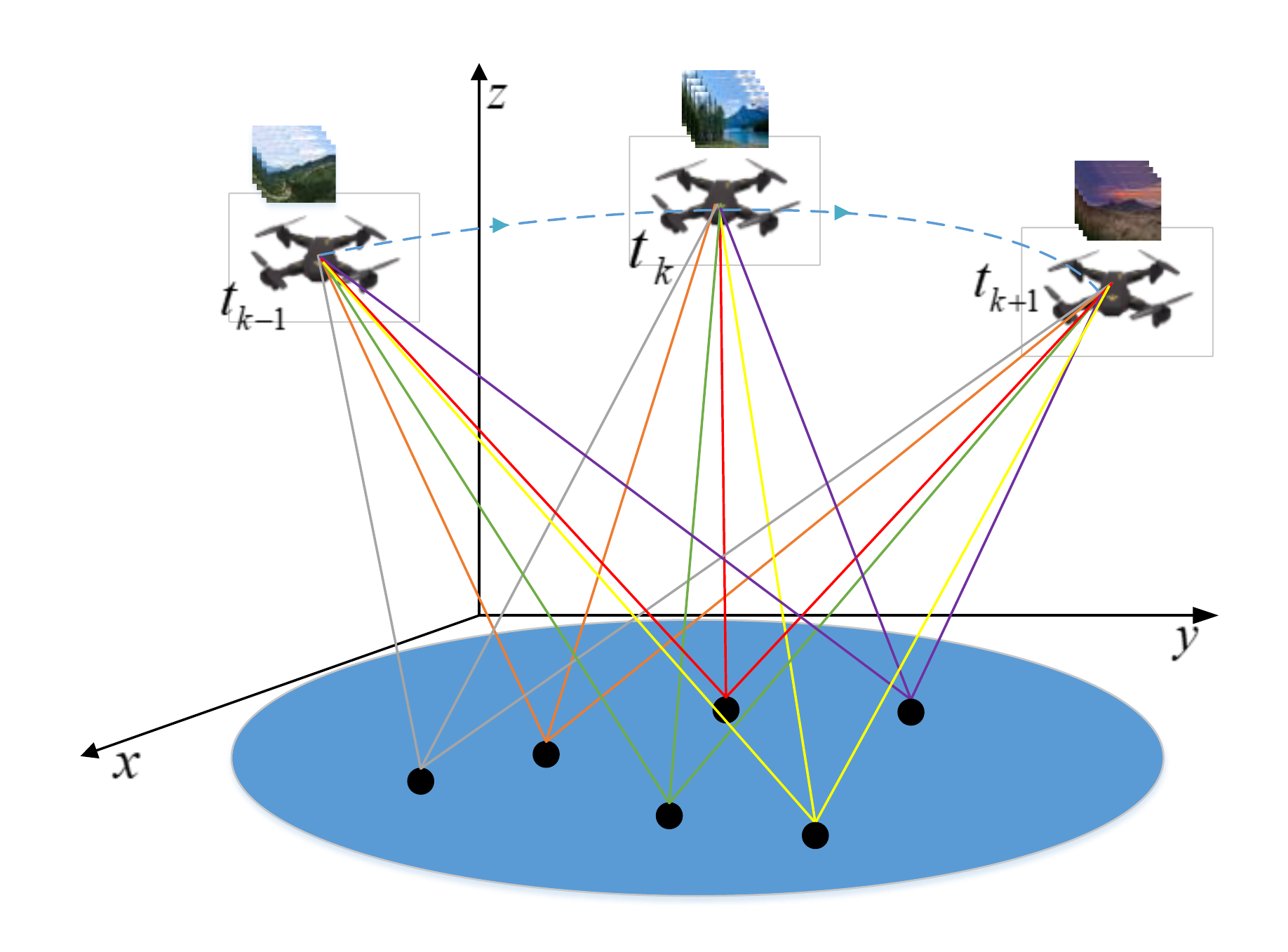}
\caption{The overview of the QUPWV-Cast system.}
\label{F1}
\end{figure}

\subsection{Distortion Estimation Model}
Denote the DCT coefficient blocks by the random variable $\left\{ {{X_m},m = 1, \cdot  \cdot  \cdot ,M} \right\}$, where $M$ represents the total number of coefficient blocks. Assume that ${X_m}\sim{\cal N}(0,{\lambda _m})$ where $\lambda _m$ represents the variance of $X_m$. In addition, we assume that ${X_m}$ is sorted in the descending order of variance, that is, ${\lambda _1} \ge \cdot\!\cdot\!\cdot\!\ge {\lambda _{{M}}}$. Let $x_k[j], 1 \le j \le N_p$ denote the $j^{th}$ DCT coefficient transmitted in time slot $k$, where $N_p$ represents the total number of coefficients contained in each block. Let $s_k$ denote the power scaling factor assigned for the coefficient block transmitted in time slot $k$. The video signal $x_k[j]$ after power scaling can be represented as $y_k[j] = s_kx_k[j]$, where $E[{\mathcal\bf{Y}_k^2}] = p_k$ represents the average transmission power allocated for coefficients in $X_k$. Obviously, we have ${p_k} = s_k^2{\lambda _k}$. Consequently, the video signal received by $u_n$ in time slot $k$, i.e., $\widetilde {y}_{n,k}[j]$, can be denoted as
\begin{equation}\label{E4}
\widetilde {{y}}_{n,k}[j] = {h_k}(n){y_{k}}[j] + {z_k}(n) = {h_k}(n){s_k}{x_{k}}[j]+ {z_k}(n),
\end{equation}
where ${h_k}(n)$ represents the instantaneous channel gain between the UAV and $u_n$ in time slot $k$, which has been given in Eq. (3). ${z_k}(n)$ is drawn \emph{i.i.d} from a zero-mean Gaussian distribution with variance $\sigma _0^2$, which can be denoted as ${z_k}(n)\sim{\cal N}\left( {0,\sigma _0^2} \right)$. As a consequence, the estimation of the demodulated video signal of $u_n$ can be denoted as
\begin{equation}\label{E5}
{\widehat{x}}_{n,k}[j] = \frac{{\widetilde {{y}}_{n,k}}[j]}{{{h_k}(n){s_k}}}  = {x_{k}}[j] + \frac{{{z_k}(n)}}{{{h_k}(n){s_k}}}.
\end{equation}

Eq. (5) indicates that the item $\frac{{{z_k}(n)}}{{{h_k}(n){s_k}}}$ determines the distortion of the transmitted video signals. In the PAVT scheme, we may discard some coefficient blocks with small variances\footnote{In natural videos, most DCT coefficients have a zero value because video frames tend to be smooth. The coefficient block with small variance usually contains little useful information.} to meet the bandwidth requirement without significantly affecting GUs' QoE [7]. The receivers can treat the discarded DCT coefficients as zeros when demodulating the video. As a consequence, the distortion caused by discarding coefficient blocks with small variances is only the sum of the squares of all discarded DCT coefficients. In this paper, we assume that only $K$ of $M$ coefficient blocks will be transmitted to GUs, and the remaining coefficient blocks will be discarded in order to meet the bandwidth requirement. For $u_n$, the expectation of the video reconstruction distortion can be denoted as follows
\begin{align}	
{D_n} &= E\left( {\sum\limits_{k = 1}^K {{\sum\limits_{j = 1}^{{N_p}} {{{( {\widehat {{x}}_{n,k}[j]\!-\! {x_{k}[j]}})}^2}} }} } \right) \!+\!\sum\limits_{m = K + 1}^M {{\sum\limits_{j = 1}^{{N_p}} {x_{m}^2[j]}} } , \notag \\
&= \sum\limits_{k = 1}^K {\frac{{{N_p}E(z_k^2(n))}}{{h_k^2(n)s_k^2}}}  + \sum\limits_{m = K + 1}^M {{\sum\limits_{j = 1}^{{N_p}} {x_{m}^2[j]}} }, \notag \\
&= {N_p}\left( {\sum\limits_{k = 1}^K {\frac{{{\sigma_0^2}{\lambda _k}}}{{h_k^2(n){p_k}}}}  + \sum\limits_{m = K + 1}^M {{\lambda _m}} } \right),\forall n.
\end{align}

\subsection{Energy Consumption Model}
The UAV's limited battery lifetime becomes the performance bottleneck of UAV-enabled video transmission systems. Therefore, it is critical to allocate the UAV's limited energy properly. In this section, we will discuss the UAV's energy consumption which is used to support the communications with GUs as well as the flight.

\subsubsection{Energy Consumed by UAV's Communications}
In the proposed QUPWV-Cast system, the UAV needs to properly allocate certain transmission power for each coefficient block, aiming at minimizing the GUs' video demodulation distortion. Then the GUs demodulate the video adaptively according to their own channel qualities. It should be noted that a small amount of side information (e.g., the mean and variance of each coefficient block) also needs to be transmitted to GUs besides the DCT coefficients [7]. The most reliable MCS is selected to ensure that each GU can decode the side information correctly. However, the energy consumed by transmitting these side information is often far less than that consumed by transmitting DCT coefficients. Specifically, the energy consumed by communications can be denoted as
\begin{equation}\label{E7}
{E_c} = {N_p}{\Delta}\sum\limits_{k = 1}^K {{p_k}},
\end{equation}
where ${p_k}$ represents the average transmission power allocated for each coefficient of the block transmitted in time slot $k$. In this paper, we use ${\bf{P}} \buildrel \Delta \over = \{ {p_k},\forall k \in {\cal K}\}$ to represent the transmission power allocation strategy. The average transmission power allocated for coefficients is not supposed to exceed the maximum average transmission power ${\overline P _{\max }}$. Therefore, the energy consumed by UAV communications is upper bounded by ${E_{\max }} = KN_p\Delta \overline P_{\max}$.

\subsubsection{Energy Consumed by UAV's flight}
Let ${\bf{v}}[k]\! \in \! {{\mathbb{R}}^{3 \times 1}}$ and ${\bf{a}}[k]\!\in\!{{\mathbb{R}}^{3 \times 1}}$ denote the UAV's velocity and acceleration in time slot $k$, respectively. We assume that ${v_{\min }}\!\le\!||{\bf{v}}[k]|| \le {v_{\max }},\forall k\!\in\!{\cal K}$ and $||{\bf{a}}[k]|| \!\le\!{a_{\max }},\forall k\!\in\!{\cal K}$, where ${v_{\max }}$ and ${a_{\max }}$ represent the UAV's maximum velocity and acceleration, respectively; and ${v_{\min }}$ represents the minimum velocity required by the fixed-wing UAV to maintain its flight. In addition, we define ${\bf{v}}[0] \buildrel \Delta \over = {{\bf{v}}_0}$ as the UAV's initial velocity and ${\bf{a}}[0] \buildrel \Delta \over =  {{\bf{a}}_0}$ as the UAV's initial acceleration.

According to the motion mechanism, we have ${\bf{v}}[k]\!-\!{\bf{v}}[k - 1]\!=\!{\bf{a}}[k - 1]{\Delta}$, and ${\bf{q}}[k]\!-\!{\bf{q}}[k-1]\!=\!{\bf{v}}[k - 1]{\Delta} + \frac{1}{2}{\bf{a}}[k-1]\Delta^2, \forall k\!\in\!{\cal K}$. Therefore, for a given UAV trajectory ${\bf{Q}} \buildrel \Delta \over = \{ {\bf{q}}[k],\forall k \in {\cal K}\}$, the UAV's velocity ${\bf{V}}  \buildrel \Delta \over = \{ {\bf{v}}[k],\forall k \in {\cal K}\}$ and acceleration ${\bf{A}} \buildrel \Delta \over = \{ {\bf{a}}[k],\forall k \in {\cal K}\} $ can be uniquely determined. Consequently, ${\bf{Q}} $, ${\bf{V}}$ and ${\bf{A}}$ are a set of coupling variables, and ${\bf{Q}} $ can be represented by ${\bf{V}}$ and ${\bf{A}}$. For each time slot $k$, the propulsion power of the fixed-wing UAV (i.e., $p_k^f$) measured in joules can be approximately represented as follows [44]
\begin{equation}\label{E8}
p_k^f = {c_1}||{\bf{v}}[k]|{|^3} + \frac{{{c_2}}}{{||{\bf{v}}[k]||}}\left( {1 + \frac{{||{\bf{a}}[k]|{|^2}}}{{{g_0^2}}}} \right),\forall k,
\end{equation}
where ${c_1}$ and ${c_2}$ are constants related to air density, drag coefficient, wing area, etc. $g_0$ is the gravitational acceleration, and its value is $9.8m/{s^2}$. Therefore, the total energy consumed by UAV's flight can be denoted as
\begin{equation}\label{E9}
{E_f} = {\Delta}\sum\limits_{k = 1}^K {p_k^f}.
\end{equation}

\subsection{Objective Function}
In this paper, PSNR is adopted as the metric to measure the GUs' QoE of the reconstructed video [53]. It is a standard objective image evaluation metric with a definition as follows
\begin{equation}\label{E10}
P\!S\!N\!R = 10{\log _{10}}\frac{{{\eta ^2}}}{{M\!S\!E}},
\end{equation}
where $\eta$ is a constant related to the pixel depth of the video\footnote{The pixel depth means the number of bits used to hold a pixel. For example, if the pixel depth is 8 bits, then $\eta  = 2^8-1 = 255$.}. $M\!S\!E$ represents the mean squared error of pixels between the reconstructed video and the original one. According to Eq. (6), the $M\!S\!{E_n}$ of ${u_n}$ can be denoted as
\begin{equation}\label{E11}
M\!S\!{E_n} = \frac{{ {\sum\limits_{k = 1}^K {\frac{{{\sigma_0^2}{\lambda _k}}}{{h_k^2(n){p_k}}}}\!+\!\sum\limits_{m = K + 1}^M\!{{\lambda _m}} }}}{M},\forall n.
\end{equation}

Combining Eq. (10) and Eq. (11), the $P\!S\!R\!N_n$ of the reconstructed video of $u_n$ can be formulated as
\begin{equation}\label{E12}
\begin{array}{l}
P\!S\!N\!{R_n}\!=\!10{\log _{10}}{\frac{{M{\eta ^2}}}{{\!{\sum\limits_{k = 1}^K {\frac{{{\sigma_0^2}{\lambda _k}{{\left\|\!{{\bf{q}}[k] - {{\bf{w}}_n}} \right\|}^2}}}{{{\beta _0}{p{_k}}}}}\!+\!\sum\limits_{m = K + 1}^M\!{{\lambda _m}} }\!}}},\forall n.
\end{array}
\end{equation}

From Eq. (12), one can see that the video demodulation quality is mainly affected by the following three factors:

1) {\emph{The UAV trajectory}}. The UAV trajectory actually determines the instantaneous channel gain ${h_k}(n)$. If the UAV is close enough to ${u_n}$, ${h_k}(n)$ will increase and the video demodulation quality will be improved accordingly. However, it is unlikely for the UAV to approach to all GUs at the same time. Thus it is necessary to design an optimal UAV trajectory to maximize the minimum PSNR of GUs.

2) {\emph{The transmission power allocation strategy}}. When more transmission power is allocated to the coefficient block, the signals decoded by GUs will be less error-prone. The transmission power allocation strategy is mainly related to the characteristics of the transmitted video data, e.g., ${\lambda _m}$. It is typical to allocate more transmission power to the coefficient blocks with large variances since those blocks may contain more useful information than those with small variances.

3) {\emph{The bandwidth budget}}. In the proposed system, some coefficient blocks need to be discarded due to the limited bandwidth resources. Eq. (12) shows that less coefficient blocks may worsen the reconstruction distortion.

\subsection{Problem Statement}
It has been pointed out in [7] that the reconstructed video with PSNR lower than 20 dB is unsuitable for high-resolution video applications. Therefore, the goal of this paper is to maximize the minimum PSNR of the video reconstructed by the GUs via jointly optimizing the transmission power allocation strategy (i.e., ${\bf{P}}$) and the UAV trajectory (i.e., ${\bf{Q}},{\bf{V}},{\bf{A}}$) under the constraints of UAV's total energy and motion mechanism. The optimization problem can be expressed as follows:
\begin{align}
(P1)\mathop {\max }\limits_{{\bf{P}},{\bf{Q}},{\bf{V}},{\bf{A}}} &\mathop {\min }\limits_n \;P\!S\!N\!{R_n},\\
{\rm{s.t.}}~~&{E_c} + {E_f} \le {E_t},\\
&0 \le {E_c} \le {E_{\max }},\\
&{v_{\min }} \le ||{\bf{v}}[k]|| \le {v_{\max }},\forall k,\\
&||{\bf{a}}[k]|| \le {a_{\max }},\forall k,\\
&{\bf{v}}[k] - {\bf{v}}[k - 1] = {\bf{a}}[k - 1]{\Delta},\forall k,\\
&{\bf{q}}[k]\!\!-\!\!{\bf{q}}[k\!-\!1]\!\!=\!\!{\bf{v}}\![k\!-\!1]\!{\Delta}\!\!+\!\!\frac{1}{2}{\bf{a}}[k\!-\!1]\!\Delta^2\!,\forall k,\\
&{\bf{q}}[0]\!\!=\!\!{{\bf{w}}_0},{\bf{q}}[K]\!\!=\!\!{{\bf{w}}_F},{\bf{v}}[0]\!\!=\!\!{{\bf{v}}_0},{\bf{a}}[0]\!\!=\!\!{{\bf{a}}_0},
\end{align}
 where ${E_t}$ represents the UAV's total available energy including the energy consumed by communications (i.e., ${E_c}$) and flight (i.e., ${E_f}$). The constraint (14) indicates that the consumed energy can not exceed the UAV's total available energy. The constraint (15) stipulates that the consumed transmission energy must be greater than $0$ but no more than the maximum allowable transmission energy ${E_{\max}}$. The constraint (16) specifies the UAV's minimum and maximum velocity. The constraint (17) specifies the UAV's maximum acceleration. The constraint (18) shows the relationship between the UAV's velocity and acceleration in each time slot. The constraint (19) shows the relationship between the UAV's coordinate and velocity/acceleration in each time slot. The constraint (20) gives the UAV's starting point, ending point, initial velocity and acceleration, respectively. Since the constraints (14) and (16) as well as the objective function in Eq. (13) are non-convex, {\emph{P}}1 becomes an intractable non-convex optimization problem.
\vspace{-4mm}
\section{Solution}
In this section, BCD and SCA techniques will be applied to obtain the sub-optimal solution to the original optimization problem {\emph{P}}1. Before solving {\emph{P}}1, we make the notation: $\mu  \buildrel \Delta \over = \mathop {\min }\limits_n P\!S\!N\!{R_n}$. After combining Eqs. (7)-(9) and (12), we can rewrite {\emph{P}}1 as follows:
\begin{align}
(P2)&\mathop {\max }\limits_{{\bf{P}},{\bf{Q}},{\bf{V}},{\bf{A}},\mu } \;\mu\\
&\rm{s.t.}~~(15)-(20), \nonumber \\
&10{\log_{10}}\!\frac{{M{\eta ^2}}}{{{\sum\limits_{k = 1}^K {\frac{{{\sigma_0 ^2}{\lambda _k}{{\left\| {{\bf{q}}[k] - {{\bf{w}}_n}} \right\|}^2}}}{{{\beta _0}{p_k}}}}\!+\!\sum\limits_{m = K + 1}^M \!{{\lambda _m}} }}} \!\ge\! \mu ,\forall n,\\
&{\Delta}\sum\limits_{k=1}^K\!{\left({{\!N_p}{p_k}\!+\!{c_1}\!||\!{\bf{v}}[k]\!|{|^3}\!+\! \frac{{{c_2}}}{{||\!{\bf{v}}[k]\!||}}\!\left({1\!+\!\frac{{||{\bf{a}}[k]|{|^2}}}{{{g_0^2}}}}\!\right)}\!\right)}\!\!\le\!\!{E_t}.
\end{align}

Since the maximization function is a convex function, the objective function of {\emph{P}}2 is now convex. We only need to consider how to transform the non-convex constraints (i.e., (16), (22), and (23)) to convex ones. Specifically, we decompose {\emph{P}}2 into two sub-optimization problems, that is 1) {\emph{transmission power allocation strategy optimization with fixed UAV trajectory}}, and 2) {\emph{UAV trajectory optimization with fixed transmission power allocation}}. Then we could use a globally iterative algorithm (discussion below) to obtain the sub-optimal solution to {\emph{P}}2.

\subsection{Transmission Power Allocation Strategy Optimization with Fixed UAV Trajectory}
Given the UAV trajectory, i.e., ${\bf{Q}}$, ${\bf{V}}$, and ${\bf{A}}$, we consider the following sub-optimization problem of {\emph{P}}2 to optimize the transmission power allocation strategy ${\bf{P}}$
\begin{align}
(P3)~~\mathop {\max }\limits_{{\bf{P}},\mu } &~\mu \\
{\rm{s.t.}}\;&(15),(22),(23). \notag
\end{align}

In the sub-optimization problem {\emph{P}}3, we introduce the following lemma to prove the convexity of the constraint (22).
\newtheorem{lemma}{\emph{\underline{Lemma}}}
\begin{lemma}
\label{lemma1}
Given the feasible UAV trajectory, i.e., ${\bf{Q}}$, ${\bf{V}}$, and ${\bf{A}}$, the constraint (22) is convex with respect to ${\bf{P}}$.
\end{lemma}
\begin{proof}
Please see \emph{Appendix A} for the proof details.
\end{proof}
The constraints (15) and (23) are also convex with respect to ${\bf{Q}}$. Consequently, the sub-optimization problem {\emph{P}}3 is a standard convex optimization problem which can be solved through some optimization toolboxes, e.g., CVX.
\vspace{-2mm}
\subsection{UAV Trajectory Optimization with Fixed Transmission Power Allocation Strategy}
Given the feasible transmission power allocation strategy ${\bf{P}}$, we consider the following sub-optimization problem of {\emph{P}}2 to optimize the UAV trajectory ${\bf{Q}}$, ${\bf{V}}$, and ${\bf{A}}$.
\begin{align}
(P4)~~\mathop {\max }\limits_{{\bf{Q}},{\bf{V}},{\bf{A}},\mu } &\;\mu \\
{\rm{s.t.}}\;\;&(16)\!-\!(20),(22),(23). \notag
\end{align}

The sub-optimization problem {\emph{P}}4 is a non-convex problem due to the non-convexity of the constraints (16), (22) and (23). Before solving the sub-optimization problem {\emph{P}}4, we first introduce a set of slack variables ${\bf{O}} = \{o[k],\forall k\}$. Note that $o[k]$ is a scalar. Then {\emph{P}}4 can be rewritten as the following optimization problem,
\begin{align}
(P5)~~&\mathop {\max }\limits_{{\bf{Q}},{\bf{V}},{\bf{A}},{\bf{O}},\mu }\;\mu\\
{\rm{s.t.}}\;\;&(17)-(20),(22),\notag \\
&||{\bf{v}}[k]|{|^2} \le v_{\max }^2,\forall k,\\
&v_{\min }^2 \le ||{\bf{v}}[k]|{|^2},\forall k,\\
&{o^2}[k] \le ||{\bf{v}}[k]|{|^2},\forall k,\\
&o[k] \ge 0,\forall k,\\
&{\Delta}\!\sum\limits_{k = 1}^K\!{\left(\!{{c_1}||{\bf{v}}[k]|{|^3}\!+\!\frac{{{c_2}}}{{o[k]}}\left( {1\!+\!\frac{{||{\bf{a}}[k]|{|^2}}}{{{g_0^2}}}}\right)} \right)} \le{E_t}\!-\!{E_c}.
\end{align}

Note that only the reciprocal term of $||{\bf{v}}[k]||$ is replaced by $o[k]$, and the cubic term of $||{\bf{v}}[k]||$ is reserved under the constraint (31). As a consequence, the left-hand side of the constraint (31) is monotonically decreasing with respect to $o[k]$.
\begin{lemma}
\label{lemma2}
Without compromising the optimality, the optimal solution to P5 must meet: $||{\bf{v}}[k]|| = o[k],\forall k$.
\end{lemma}
\begin{proof}
Please see Appendix B.
\end{proof}

According to Lemma 2, we can conclude that the solution of {\emph{P}}5 is the same as {\emph{P}}4 due to the equivalence between the constraints (27)-(28) in {\emph{P}}5 and the constraint (16) in {\emph{P}}4. After introducing the slack variable set ${\bf{O}}$, the constraint (31) is now convex with respect to ${\bf{V}}$ and ${\bf{A}}$. However, {\emph{P}}5 is still non-convex due to the non-convexity of the constraints (22), (28) and (29).

To address this challenge, we propose to derive a sub-optimal solution to {\emph{P}}5 by applying the SCA technique. Specifically, we can derive the convex approximation of the right-hand side of the constraints (28)-(29), i.e., $||{\bf{v}}[k]|{|^2}$, and the left-hand side of the constraint (22), i.e., $10{\log_{10}}\!\frac{{M{\eta ^2}}}{{ {\sum\limits_{k = 1}^K {\frac{{{\sigma_0 ^2}{\lambda _k}{{\left\| {{\bf{q}}[k] - {{\bf{w}}_n}} \right\|}^2}}}{{{\beta _0}{p_k}}}}\!+\!\sum\limits_{m = K + 1}^M \!{{\lambda _m}} }}}$, at a given point in each iteration.

In the constraints (28)-(29), $||{\bf{v}}[k]|{|^2}$ is convex with respect to ${\bf{v}}[k]$. Since the first-order Taylor approximation of a convex function is a global under-estimator, $||{\bf{v}}[k]|{|^2}$ can be lower bounded as follows
\begin{equation}\label{E32}
||{\bf{v}}[k]|{|^2} \ge ||{{\bf{v}}^r}[k]|{|^2} + 2{\left( {{{\bf{v}}^r}[k]} \right)^T}\left( {{\bf{v}}[k] - {{\bf{v}}^r}[k]} \right),\forall k,
\end{equation}
where ${{\bf{v}}^r}[k]$ represents the given approximation point in the $r^{th}$ iteration. The equality holds at the point ${\bf{v}}[k] = {\bf{v}}^r[k]$. According to Eq.(32), the constraints (28)-(29) can be rewritten as
\begin{align}
v_{\min }^2 \le ||{{\bf{v}}^r}[k]|{|^2} + 2{\left( {{{\bf{v}}^r}[k]} \right)^T}\left( {{\bf{v}}[k] - {{\bf{v}}^r}[k]} \right),\forall k,\\
{o^2}[k] \le ||{{\bf{v}}^r}[k]|{|^2} + 2{\left( {{{\bf{v}}^r}[k]} \right)^T}\left( {{\bf{v}}[k] - {{\bf{v}}^r}[k]} \right),\forall k.
\end{align}

In Lemma 1, we have proved that the left-hand side of the constraint (22) is a concave function with respect to ${\bf{P}}$. Therefore, it can be inferred from the proof in \emph{Appendix A} that the left-hand side of the constraint (22) can not be a concave function with respect to ${\bf{Q}}$. For simplicity, we omit the proof process since it is similar to Lemma 1. In order to solve the non-convexity of the constraint (22), we still use the first-order Taylor approximation to obtain the lower bound of the left-hand side of the constraint (22) as follows
\begin{align}
&P\!S\!N\!{R_n}({\bf{Q}}) \ge P\!S\!N\!{R_n}({{\bf{Q}}^r}) \notag \\
&\!=\!I_n\!-\!\sum\limits_{k = 1}^K\!{{J_n}[k]}\!\left({||{\bf{q}}[k]\!-\!{{\bf{w}}_n}|{|^2}\!-\!||{{\bf{q}}^r}[k]\!-\!{{\bf{w}}_n}|{|^2}}\right)\!,\forall n,
\end{align}
where ${{\bf{Q}}^r} \buildrel \Delta \over = \{ {{\bf{q}}^r}[k],\forall k\}$ is defined as the given feasible UAV trajectory in the $r^{th}$ iteration. The equality holds at the point ${\bf{q}}[k] = {\bf{q}}^r[k]$. The coefficients ${I_n}$ and ${J_n}[k]$ can be denoted as
\begin{align}
&I_n = 10{\log _{10}}\!\frac{{M{\eta ^2}}}{{{\sum\limits_{k = 1}^K {\frac{{{\sigma_0 ^2}{\lambda _k}{{\left\| {{{\bf{q}}^r}[k]\!-\!{{\bf{w}}_n}} \right\|}^2}}}{{{\beta _0}{p_k}}}}  + \sum\limits_{m = K + 1}^M {{\lambda _m}} } }},\forall n,\\
&{J_n}[k]\!=\!\frac{{10 \cdot \frac{{{\sigma ^2}{\lambda _k}}}{{{\beta _0}{p_k}}}}}{{\ln\!10\!\left(\!{\sum\limits_{k = 1}^K \!{\frac{{{\sigma_0^2}\!{\lambda _k}\!{{\left\| {{{\bf{q}}^r}[k]\!-\!{{\bf{w}}_n}} \right\|}^2}}}{{{\beta _0}{p_k}}}}\!+\!\!\sum\limits_{m = K + 1}^M {{\lambda _m}} }\!\right)}},\forall k,n.
\end{align}

According to Eqs. (35)-(37), the constraint (22) in {\emph{P}}5 can be rewritten as follows
\begin{equation}\label{E38}
I_n\!-\!\sum\limits_{k = 1}^K\!{{J_n}[k]}\!\left( {||{\bf{q}}[k]\!-\!{{\bf{w}}_n}|{|^2}\!-\!||{{\bf{q}}^r}[k]\!-\!{{\bf{w}}_n}|{|^2}}\!\right) \ge \mu ,\forall n.
\end{equation}

Since the left-hand side of Eq. (38) is a concave function with respect to $||{\bf{q}}[k] - {{\bf{w}}_n}|{|^2}$, the constraint (38) is now a convex constraint. By combining Eqs. (33), (34) and (38), {\emph{P}}5 can be reformulated as follows
\begin{align}
(P6)~~\mathop {\max }\limits_{{\bf{Q}},{\bf{V}},{\bf{A}},{\bf{O}},\mu } &\;\mu \\
{\rm{s.t.}}~~~&(17)-(20),(27),(30),(31),(33),(34),(38).  \notag
\end{align}

In {\emph{P}}6, all constraints now satisfy the convexity requirements. Therefore, {\emph{P}}6 is a standard convex optimization problem, which can be efficiently solved by existing solvers, e.g., CVX. It should be noticed that, due to the lower bounds given in Eqs. (32) and (35), the constraints (22), (28) and (29) in {\emph{P}}5 will be satisfied as the constraints (33), (34) and (38) in {\emph{P}}6 are satisfied, but not \emph{vice versa}. Therefore, the feasible solution to {\emph{P}}6 is a subset of {\emph{P}}5.

\subsection{Complexity and Convergence Analysis}
Based on the results of the above two sub-optimization problems, the overall algorithm for computing the sub-optimal solution to {\emph{P}}2 is summarized in \emph{Algorithm 1}. The complexity of \emph{Algorithm 1} is analyzed as follows. In each iteration, the transmission power allocation strategy ${\bf{P}}$ and the UAV trajectory ${\bf{Q}}$ are iteratively optimized using the convex solver based on the interior-point method, and thus their individual complexity can be represented as $O((K)^{3.5}\rm{log}(1/\varsigma))$ and $O((3K)^{3.5}\rm{log}(1/\varsigma))$, respectively. Specifically, $K$ represents the total time slots and $\varsigma$ represents the predetermined solution accuracy. Then accounting for the BCD iterations with the complexity in the order of $\rm{log}(1/\varsigma)$, the total computation complexity of \emph{Algorithm 1} is $O((3K)^{3.5}\rm{log}^2(1/\varsigma))$ [52].
\begin{algorithm}[htb]
\caption{Iterative Optimization for {\bf{P}} and {\bf{Q}}.}
\label{alg:Framework}
\begin{algorithmic}[1]
\State Initialize ${{\bf{P}}^0}$, ${{\bf{Q}}^0}$, ${{\bf{V}}^0}$, and ${{\bf{A}}^0}$.
\State Set up a convergence threshold $\varsigma  > 0$.
\State Let $r = 0$.
\State {\bf{repeat}}
\State Solve {\emph{P}}3 with given UAV trajectory ${{\bf{Q}}^r}$, ${{\bf{V}}^r}$, and ${{\bf{A}}^r}$, and denote the optimal solution to the transmission power allocation strategy as ${{\bf{P}}^{r+1}}$.
\State Solve {\emph{P}}6 with given transmission power allocation strategy ${{\bf{P}}^{r+1}}$, and denote the optimal solution to the UAV trajectory as ${{\bf{Q}}^{r+1}}$, ${{\bf{V}}^{r+1}}$, and ${{\bf{A}}^{r+1}}$.
\State Update $r = r+1$.
\State {\bf{Break}}: if $\frac{{\mu}(\!{{\bf{P}}^{r\!+\!1}},{{\bf{Q}}^{r\!+\!1}},{{\bf{V}}^{r\!+\!1}},{{\bf{A}}^{r+1}})-{\mu}(\!{{\bf{P}}^{r}},{{\bf{Q}}^{r}},{{\bf{V}}^{r}},{{\bf{A}}^{r}})}{{\mu}(\!{{\bf{P}}^{r}},{{\bf{Q}}^{r}},{{\bf{V}}^{r}},{{\bf{A}}^{r}})}  \le \varsigma$.
\end{algorithmic}
\end{algorithm}

Next, we investigate the convergence property of \emph{Algorithm 1}. Let $\mu({{\bf{P}}^r},{{\bf{Q}}^r},{{\bf{V}}^r},{{\bf{A}}^r})$ denote the objective value of {\emph{P}}6 in the $r^{th}$ iteration. Therefore, the following inequality holds,
\begin{equation}\label{E40}
\begin{array}{l}
\mu ({{\bf{P}}^r},{{\bf{Q}}^r},{{\bf{V}}^r},{{\bf{A}}^r})\mathop  \le \limits^{(a)} \mu ({{\bf{P}}^{r + 1}},{{\bf{Q}}^r},{{\bf{V}}^r},{{\bf{A}}^r})\\
\!\mathop \le\limits^{(b)}\!{\mu }(\!{{\bf{P}}^{r\!+\!1}}\!,\!{{\bf{Q}}^{r\!+\!1}}\!,\!{{\bf{V}}^{r\!+\!1}}\!,\!{{\bf{A}}^{r\!+\! 1}}\!)\!\mathop  \le \limits^{(c)}\!\mu^*\!({{\bf{P}}^{r\!+\!1}}\!,\!{{\bf{Q}}^{r\!+\!1}}\!,\!{{\bf{V}}^{r\!+\!1}}\!,\!{{\bf{A}}^{r\!+\!1}}\!)\!.
\end{array}
\end{equation}
where ${\mu^*}(\!{{\bf{P}}^{r\!+\!1}},{{\bf{Q}}^{r\!+\!1}},{{\bf{V}}^{r\!+\!1}},{{\bf{A}}^{r\!+\! 1}})$ represents the optimal solution to {\emph{P}}2. The inequality (a) holds since Step 5 in \emph{Algorithm 1} can obtain the optimal solution to {\emph{P}}3. The inequality (b) holds as Step 6 in \emph{Algorithm 1} can obtain the optimal solution to {\emph{P}}6. Since the SCA technique is used to achieve the lower bounds of the constraints (32) and (35), the optimal solution to {\emph{P}}6 must be the lower bound of the optimal solution to {\emph{P}}2. Consequently, the inequality (c) holds. Therefore, the optimal solutions to {\emph{P}}3 and {\emph{P}}6 are guaranteed to be non-decreasing according to Eq. (40) over the iterations. Thus \emph{Algorithm 1} can converge to a locally optimal solution to {\emph{P}}2.

\section{Performance Analysis}
We have conducted simulations to verify the effectiveness of the proposed system. Particularly we investigate the influence of three factors, e.g., $K$, $E_t$, and $N$, on the system performance. For each case, we present the simulation results about 1) the transmission power allocation strategy, 2) the UAV trajectory, 3) the convergence of the proposed algorithm, and 4) the PSNRs of GUs, respectively. Finally, we compare the performance of the proposed system with three other systems, e.g., DVB, SoftCast, and SharpCast, from the perspectives of subjective visual quality and objective evaluation metric.

\subsection{Parameter Settings}
We assume that the UAV flies at a fixed altitude as $H = 100m$. The average noise power is $\sigma_0^2=-109 dBm$ and the maximum average transmission power for each coefficient is set to $10 dBm$. $\beta_0$ is assumed to be $-40 dB$. According to the energy-consumption model of the fixed-wing UAV [43], $c_1$ and $c_2$ are set to $9.26 \times {10^{ - 4}}$ and $2.25\times {10^{3}}$, respectively. To ensure the safety, the UAV's velocity is restricted to be within the range of $3 m/s$ to $100 m/s$ and the UAV's maximum acceleration is set to $a_{\max}=10m/s^2$.

In the simulations, we require that the UAV flies from the starting point (i.e., ${{\bf{w}}_0} = {[0,300,100]^Tm}$) to the ending point (i.e., ${{\bf{w}}_F} = {[300,0,100]^Tm}$). The initial UAV trajectory is designed to fly straight from the starting point to the ending point at a constant speed. Therefore, the initial UAV trajectory can be denoted by the set ${\bf{Q}}^0= \{{\bf{q}}[k]={{{\bf{w}}_0}+\frac{k}{K}({\bf{w}}_F-{\bf{w}}_0)},\forall k\!\in\!{\cal K}\}$. Consequently, the initial feasible velocity set for the UAV is ${\bf{V}}^0=\{{\bf{v}}[k]= \frac{{\bf{w}}_F-{\bf{w}}_0}{K\Delta},\forall k \in {\cal K}\}$. The initial acceleration set for the UAV is ${\bf{A}}^0=\{{\bf{a}}[k]= [0,0,0]^T,\forall k \in {\cal K}\}$. For simplicity, the UAV's initial velocity ${\bf{v}}_0$ is also set to $\frac{{\bf{w}}_F-{\bf{w}}_0}{K\Delta}$, and the UAV's initial acceleration ${\bf{a}}_0$ is set to $[0,0,0]^Tm/s^2$, respectively. In addition, we assume that each transmitted coefficient is allocated with the maximum average transmission power at the beginning of the iteration, i.e., $p_k = 10dBm, \forall k \in\!{\cal\!K}$. Thus the initial transmission power allocation strategy can be denoted as ${\bf{P}}^0=\{p_k = 10dBm,\forall k \in {\cal K}\}$. For ease of reference, all of the parameters are provided in \textbf{Table I}.
\begin{table*}[htbp!]
\centering
\caption{The parameter settings.}
\begin{tabular*}{17.1cm}{cc|cc}
\hline
Parameter &Value &Parameter &Value\\
\hline
UAV's altitude &$H=100 m$ &Number of transmitted blocks &$K=120/140/160/180$  \\
\hline
{UAV's starting point} &${{\bf{w}}_0} = {[0,300,100]^Tm}$ &Number of coefficients in each block &$N_p=396$  \\
\hline
{UAV's ending point}  &${{\bf{w}}_F} = {[300,0,100]^Tm}$  &Maximum average transmission power &${\overline P _{\max }}=10 dBm$    \\
\hline
UAV's maximum velocity  &$v_{\max} = 100 m/s$ &Length of each time slot &$\Delta= 0.1s$ \\
\hline
UAV's minimum velocity  &$v_{\min}= 3 m/s$  &UAV's total energy  &$E_t= 3000/4000/5000/6000J$  \\
\hline
UAV's maximum acceleration &$a_{\max} = 10 m/s^2$ &Index of transmitted frames &$Index= 1^{st}, 150^{th}, 300^{th} $ \\
\hline
UAV's initial velocity  &${{\bf{v}}_0} = \frac{{\bf{w}}_F-{\bf{w}}_0}{K\Delta} m/s$ &Number of GUs &$N=4/6/8/10$  \\
\hline
UAV's initial acceleration  &${\bf{a}}_{0} = [0,0,0]^T m/s$  &Constant 1 for UAV's energy-consuming model  &${c_1} = 9.26 \times {10^{ - 4}}$   \\
\hline
Noise power &$\sigma _0^2 =  - 109dBm$ &Constant 2 for UAV's energy-consuming model &$c_2=2250$    \\
\hline
Referenced channel gain &${\beta _0} =  - 40dB$ &Gravity acceleration &$g_0 = 9.8m/{s^2}$    \\
\hline
Path loss exponent &$\alpha  = 2$ &Constant related to PSNR formulation &$\eta  = 255$   \\
\hline
Total block number &$M = 192$ &Convergence threshold &$\varsigma  = {10^{ - 4}}$   \\
\hline
\end{tabular*}
\label{T1}
\end{table*}

Six classic testing video sequences including ``Foreman'', ``Akiyo'', ``Coastguard'', ``Container'', ``Hall'', and ``Mother-daughter'' are used to verify the effectiveness of the proposed system. These video sequences are open source and have been widely used for simulation analysis in multimedia research [53]. The details of these six video sequences are provided in \textbf{Table II}, including resolution, frame number, frame rate, and size. All of them have a 8-level pixel depth (i.e., the pixel value ranges from 0 to 255). Similar to SoftCast, we divide each frame into blocks with uniform size of $22 \times 18$. Therefore, each frame contains 64 coefficient blocks and each coefficient block contains 396 DCT coefficients. For simplicity, we select the $1^{st}$, $150^{th}$, and $300^{th}$ frames from each video sequence as the signal source.
\begin{table}[htbp!]
\centering
\caption{Details of testing video sequences.}
\begin{tabular*}{8.9cm}{ccccc}
\hline
Sequences  &Resolution &Frame number & Frame rate& Size \\
\hline
Foreman &$176 \times 144$ &300& 25fps& 5.88MB \\
\hline
Akiyo &$176 \times 144$ &300& 25fps& 1.86MB \\
\hline
Coastguard &$176 \times 144$ &300& 25fps& 5.74MB \\
\hline
Container &$176 \times 144$ &300& 25fps& 4.06MB \\
\hline
Hall &$176 \times 144$ &300& 25fps& 5.72MB \\
\hline
Mother-daughter &$176 \times 144$ &300& 25fps& 10.8MB \\
\hline
\end{tabular*}
\label{T2}
\end{table}

\subsection{Influence of $K$ on the performance of the proposed scheme}
In DVB, GUs can only request the BS for the videos with a bit rate matching their channel bandwidth. It has been shown that the video quality can be modelled as a logarithmic function with respect to the bit rate [49]. Therefore, GUs with low channel quality may suffer from low QoE when requesting videos from the BS. The PAVT scheme can perform efficiently even when the channel bandwidth is limited, since it is allowed to discard some coefficient blocks with small variances, in order to meet the bandwidth budget while causing little impact on the GUs' QoE.

Here we study the impact of available channel bandwidth on the performance of the proposed scheme. We only change the value of $K$ while maintaining other parameters constant. Specifically, we set the UAV's total energy $E_t$ to $3000J$. Assume that there are four GUs which are randomly generated by Monte-Carlo method. Specifically, both the x-coordinates and y-coordinates of these GUs range from 0m to 1200m. The performance of the proposed system is investigated under four cases: $K = \{120,140,160,180\}$. The simulation results are shown in Fig. 2.
\begin{figure*}
\begin{minipage}{0.49\linewidth}
  \centerline{\includegraphics[width=8cm]{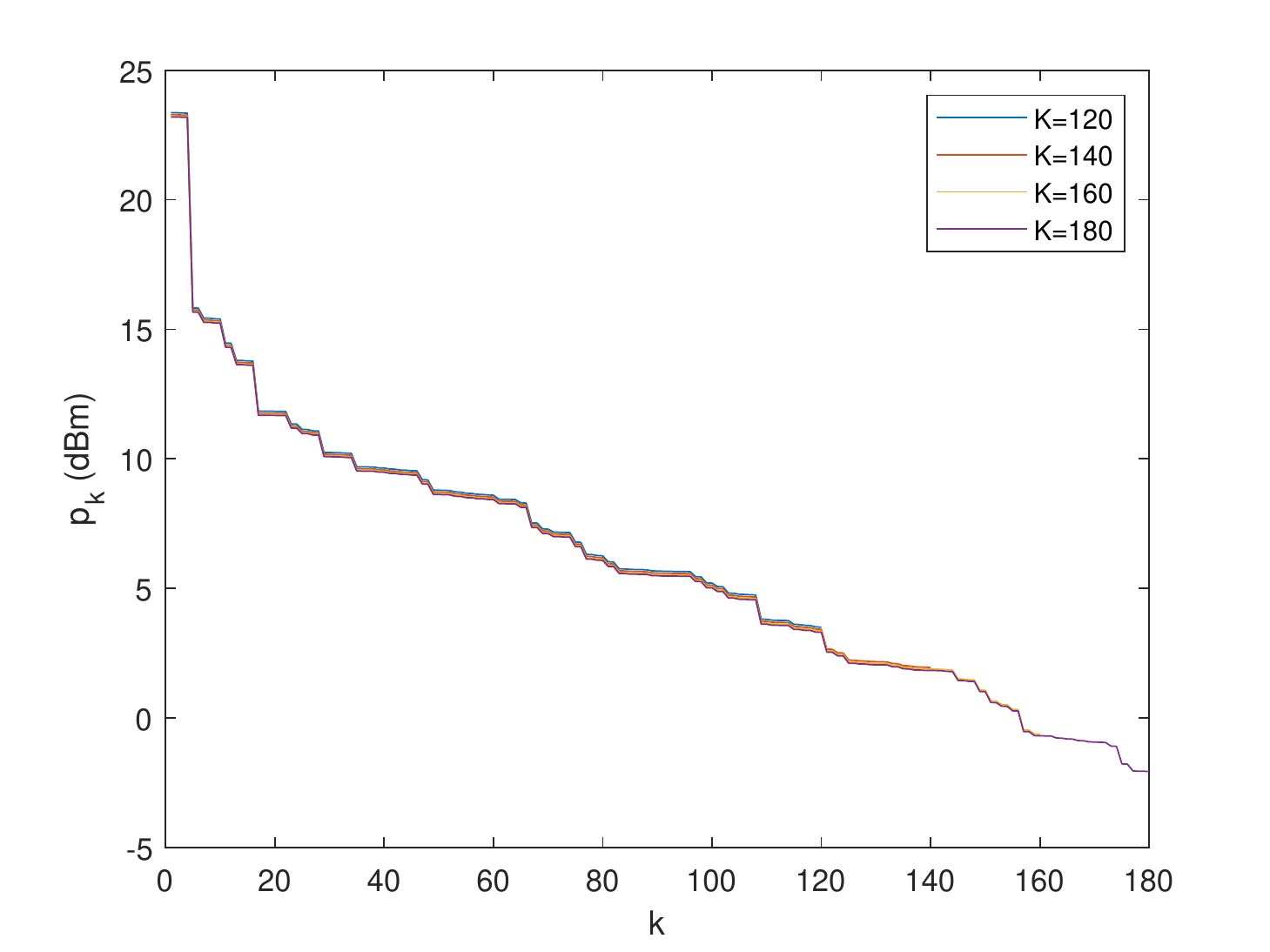}}
  \centerline{(a)}
\end{minipage}
\hfill
\begin{minipage}{0.49\linewidth}
  \centerline{\includegraphics[width=8cm]{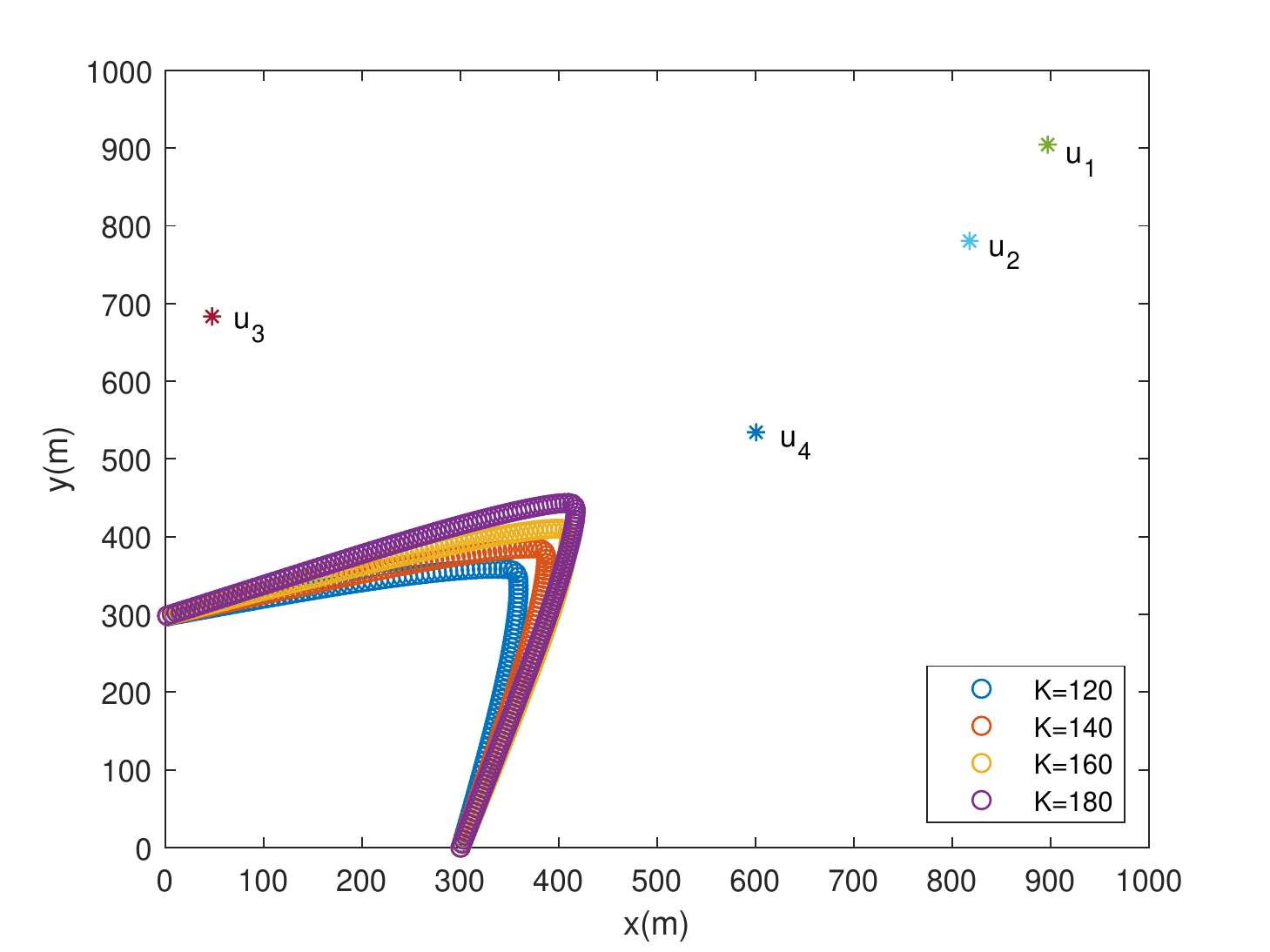}}
  \centerline{(b)}
\end{minipage}
\vfill
\begin{minipage}{0.49\linewidth}
  \centerline{\includegraphics[width=8cm]{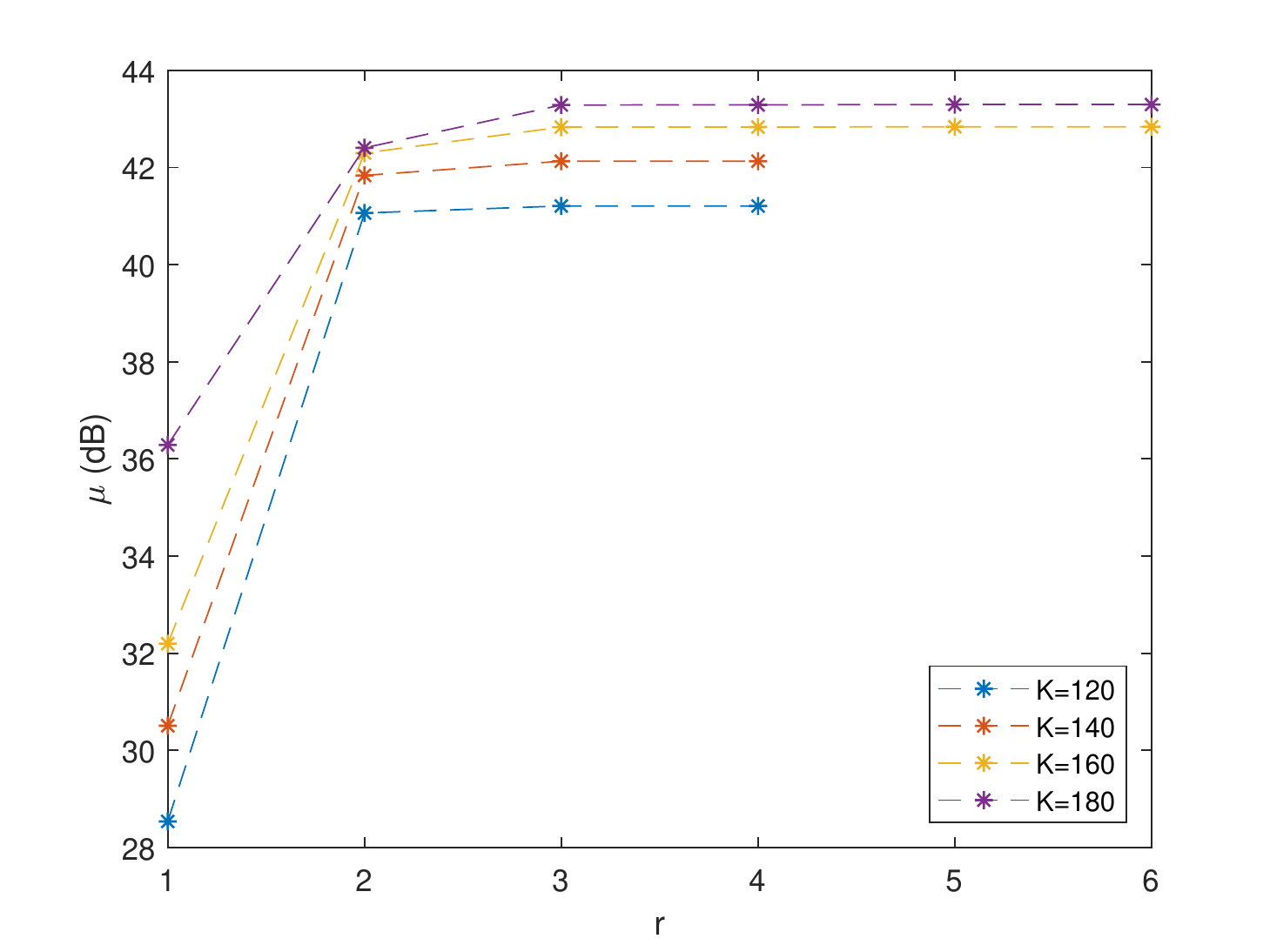}}
  \centerline{(c)}
\end{minipage}
\hfill
\begin{minipage}{0.49\linewidth}
  \centerline{\includegraphics[width=8cm]{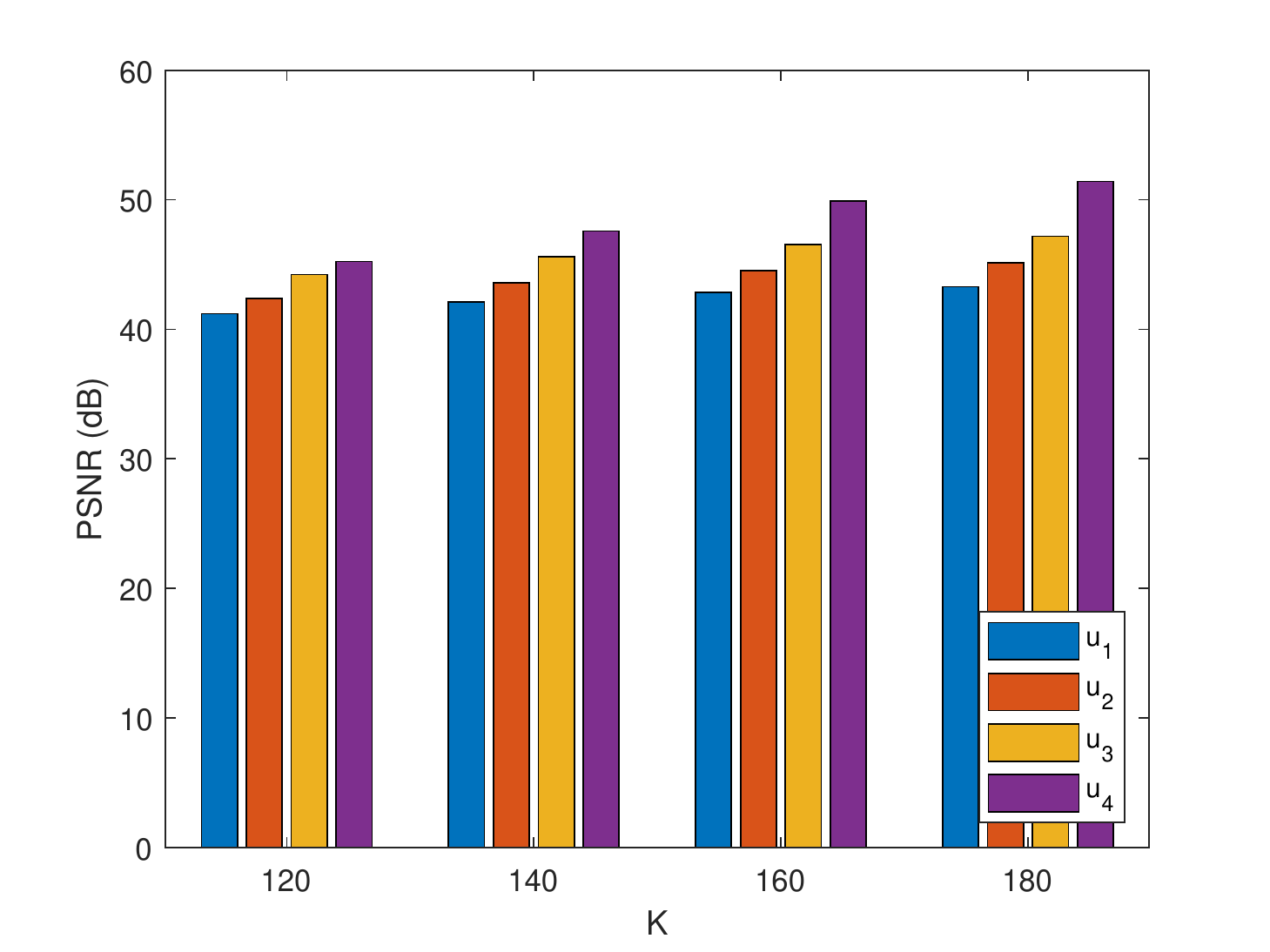}}
  \centerline{(d)}
\end{minipage}
\caption{Performance comparisons under conditions of different channel bandwidth: (a) transmission power allocation strategy, (b) UAV trajectory, (c) convergence of the proposed algorithm, and (d) PSNRs of GUs.}
\label{fig:res}
\end{figure*}

Fig. 2(a) shows the transmission power allocation strategies under different values of $K$. One can see that the average transmission power allocated for the coefficients in each block decreases with the decline of the variance (here the coefficient blocks are sorted in descending order of variance). It has been proved in SoftCast scheme that the optimal transmission power for the coefficients in each block is approximately proportional to the standard deviation. Therefore, the proposed algorithm can obtain a sub-optimal power allocation strategy as in SoftCast. However, the location relationship between the BS and GUs are not considered in SoftCast when allocating the transmission power. In this paper, the UAV can adjust the transmission power allocation according to the channel quality of GUs in each time slot. Therefore, our proposed transmission power allocation scheme is more practical and effective. Fig. 2(b) shows the UAV trajectory under four cases. We can see that the UAV approaches the GU in the furthest location (i.e., $u_1$) to maximize the video demodulation quality of the GU. Since the UAV's flight duration is limited by $K$, it can be seen from Fig. 2(b) that the UAV's flight distance increases with the increment of $K$.

Fig. 2(c) shows the convergence of the proposed algorithm. We can see that the proposed algorithm has a fast convergence speed. We have shown in Section IV that the complexity of the proposed algorithm is $O((K)^{3.5}\rm{log}(1/\varsigma))$ which scales with $K$. From Fig. 2(c), one can see that with the increase of $K$, the algorithm needs more iterations to achieve the predetermined convergence accuracy. In addition, as $K$ increases, the GUs' minimum video demodulation quality keeps improving. However, the PSNR gain with the increase of $K$ becomes limited, which indicates that the proposed system can achieve satisfactory performance even under the condition of limited bandwidth. Fig. 2(d) shows the video demodulation qualities of four GUs under different values of $K$. It can be seen that the PSNR of four GUs under each case can achieve small improvement with the increase of $K$. Specifically, the PSNR of $u_1$ (which suffers from the worst demodulation quality due to the largest distance from the UAV trajectory) can also get improved with the increase of $K$.

\subsection{Influence of $E_t$ on the performance of the proposed scheme}
It has been analysed in Section III that the video demodulation qualities of GUs are mainly determined by the following three aspects: 1) \emph{UAV transmission power allocation strategy, i.e., ${\bf{P}}$}, 2) \emph{UAV trajectory}, i.e., ${\bf{Q}}$, and 3) \emph{Channel bandwidth}, i.e., $K$. The influence of $K$ on the performance of the proposed system has been analysed. Since both ${\bf{P}}$ and ${\bf{Q}}$ are constrained by $E_t$, we will study the influence of $E_t$ on the performance of the proposed system in this section. Again, we assume that there are four GUs with coordinates the same as before. $K$ is set to 180 here. The performance of the proposed system is studied under four cases of ${E_t} = \{3000, 4000, 5000, 6000\}J$. The simulation results are given in Fig. 3.
\begin{figure*}
\begin{minipage}{0.5\linewidth}
  \centerline{\includegraphics[width=8cm]{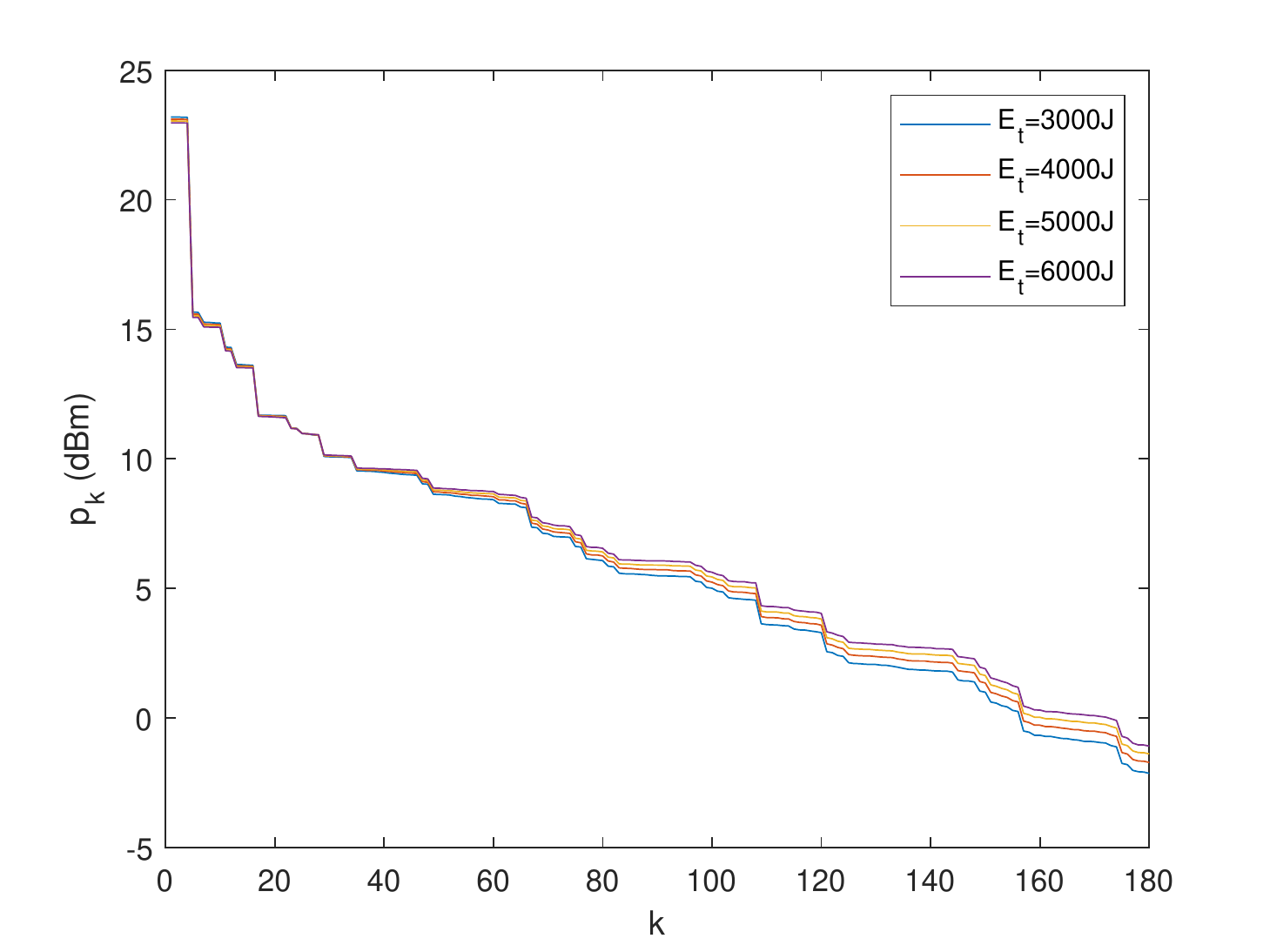}}
  \centerline{(a)}
\end{minipage}
\hfill
\begin{minipage}{0.5\linewidth}
  \centerline{\includegraphics[width=8cm]{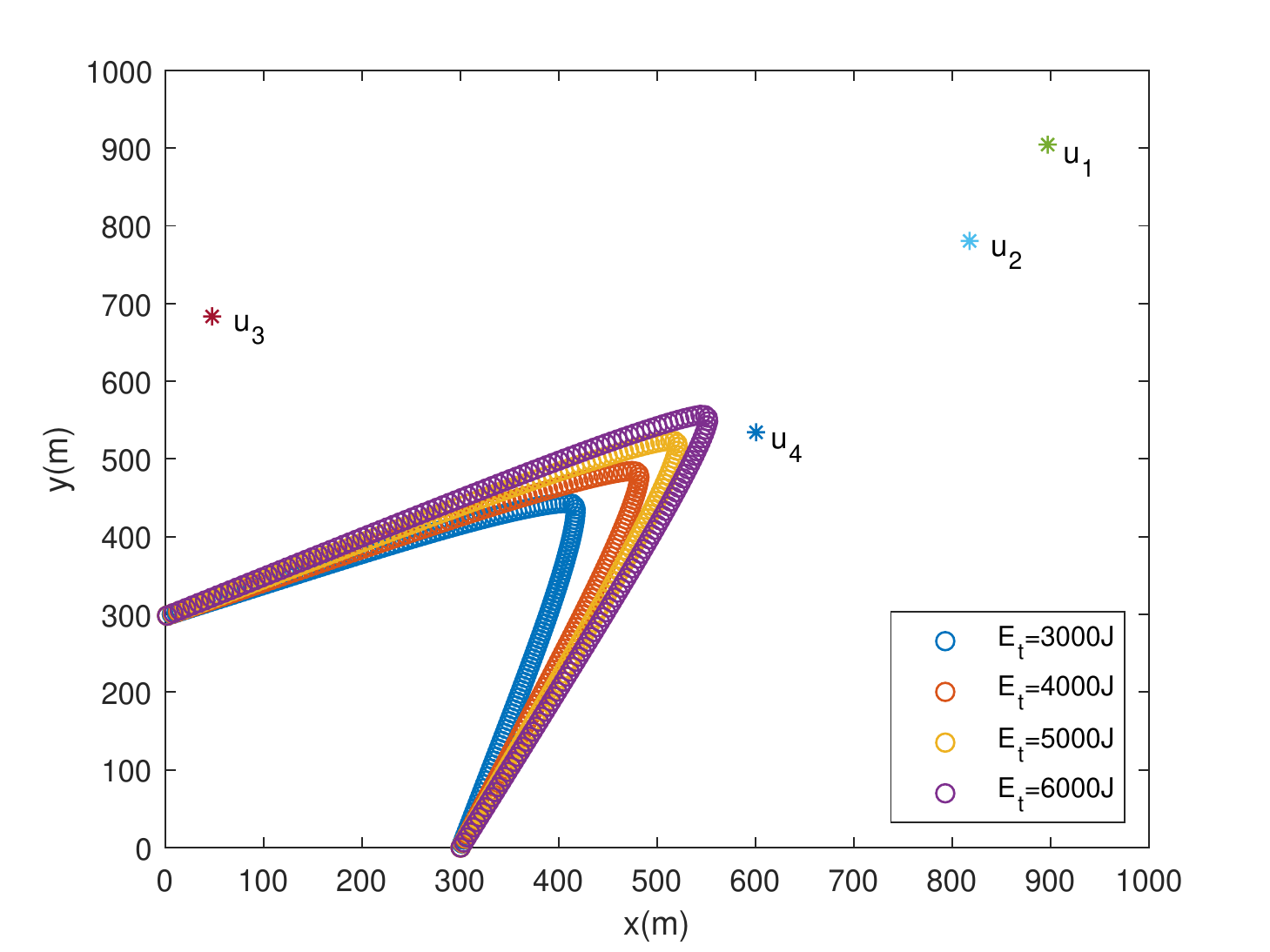}}
  \centerline{(b)}
\end{minipage}
\vfill
\begin{minipage}{0.5\linewidth}
  \centerline{\includegraphics[width=8cm]{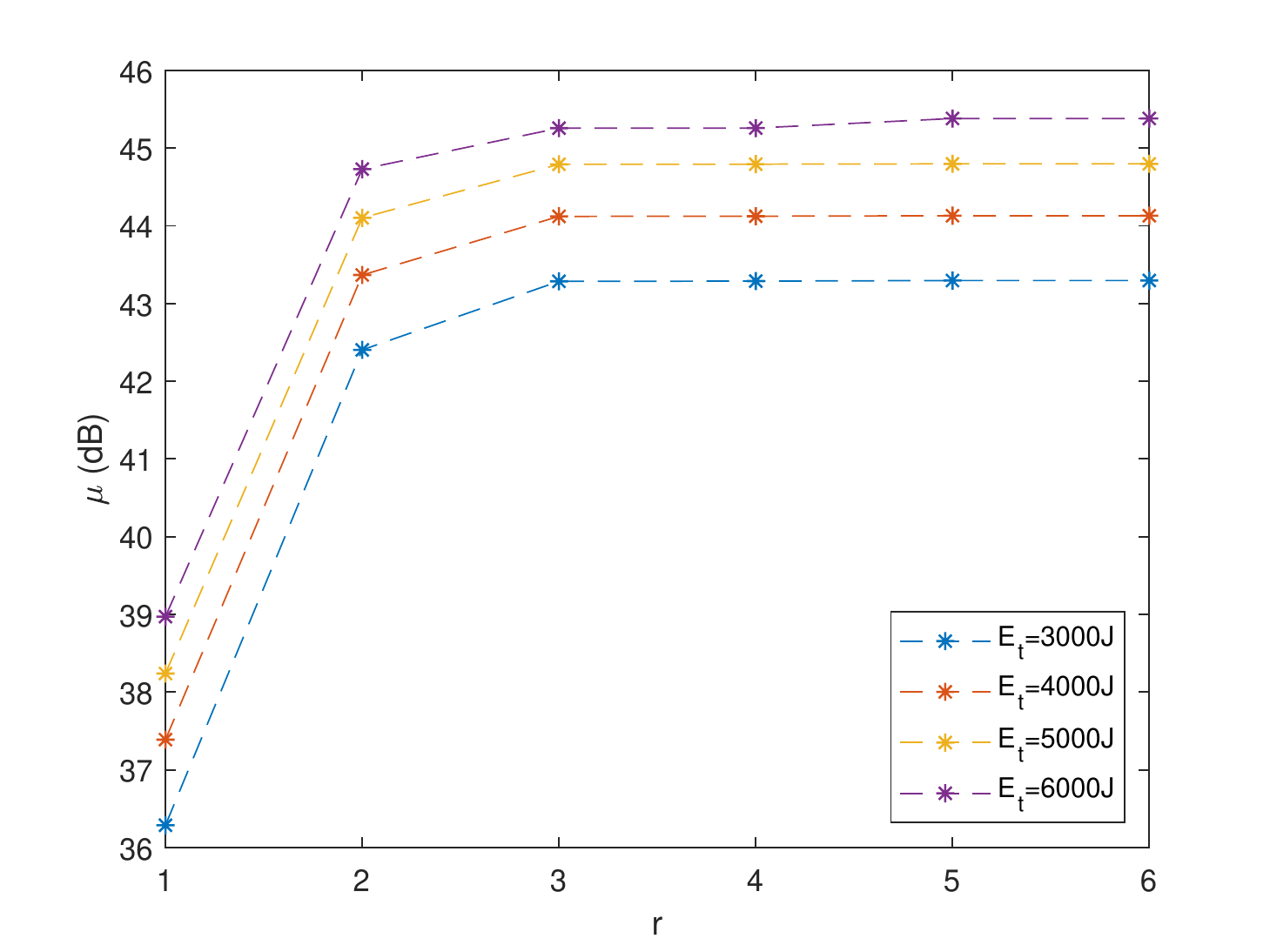}}
  \centerline{(c)}
\end{minipage}
\hfill
\begin{minipage}{0.5\linewidth}
  \centerline{\includegraphics[width=8cm]{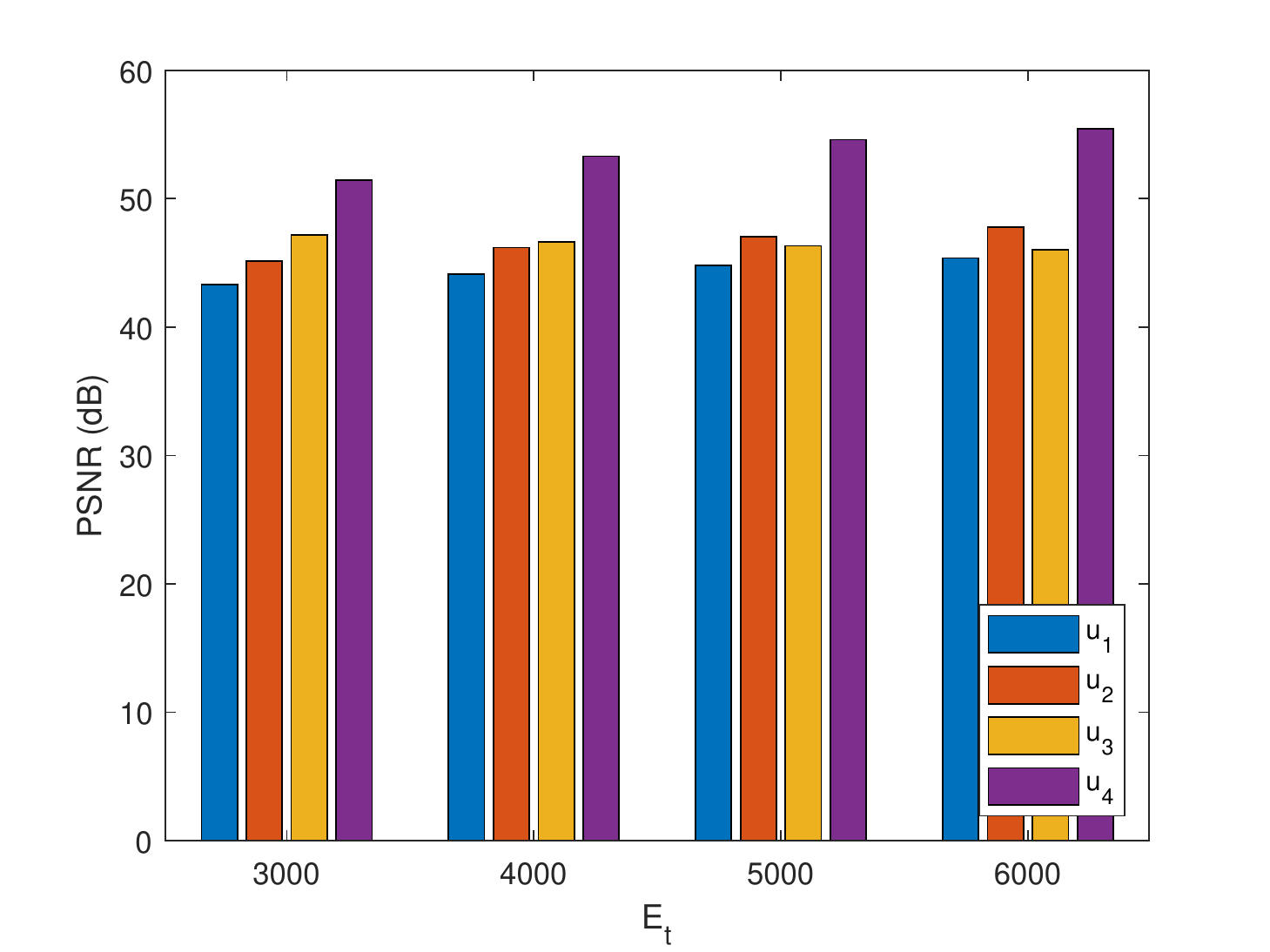}}
  \centerline{(d)}
\end{minipage}
\caption{Performance comparisons under conditions of different total energy: (a) transmission power allocation strategy, (b) UAV trajectory, (c) convergence of the proposed algorithm, (d) PSNRs of GUs.}
\label{fig:res}
\end{figure*}

Fig. 3(a) shows the transmission power allocation strategies under different cases of $E_t$. We can see that the transmission power allocated for those coefficient blocks with large variances is almost the same in all cases. As the total energy increases, the coefficient blocks with small variances can obtain more transmission power. Fig. 3(b) shows the UAV trajectory with different values of $E_t$. It shows that as the UAV's total energy increases, the UAV can fly closer to the farthest GU to maximize the video demodulation quality of the GU. In addition, one can also find that with the increase of $E_t$, the UAV's turning radius is getting smaller and smaller, which means consuming more propulsion energy\footnote{According to [44], the energy consumed by the UAV's flight is approximately proportional to the square of the UAV's acceleration.}.

Fig. 3(c) shows the convergence of the proposed algorithm under different values of $E_t$. With the increase of UAV's total energy, the sub-optimal solution obtained by the proposed algorithm is also increasing, and the increasing amplitude is approximately proportional to the energy increment. Fig. 3(d) shows the video demodulation qualities of the four GUs. From Fig. 3(d), one can conclude that with the increase of $E_t$, the minimum video demodulation quality of GUs (i.e., $u_1$) can be improved to some extent. Meanwhile, the video demodulation quality of $u_3$ gets lower since the distance is getting larger from the UAV.

\subsection{Influence of $N$ on the performance of the proposed scheme}
Besides the total energy $E_t$, the UAV trajectory is also affected by the number of GUs. The goal of the proposed system is to maximize the minimum PSNR of GUs. Therefore, we should consider the distribution of GUs when designing the UAV trajectory. We set $K$ to $180$ and $E_t$ to $3000J$, respectively. Assume that there are $10$ GUs which are randomly generated by Monte-Carlo method. Specifically, both the x-coordinate and y-coordinate of these GUs range from 0m to 1200m. We study the performance of the proposed system under the following four cases: $N = \{4, 6, 8, 10\}$. The corresponding simulation results are shown in Fig. 4.
\begin{figure*}
\begin{minipage}{0.49\linewidth}
  \centerline{\includegraphics[width=8cm]{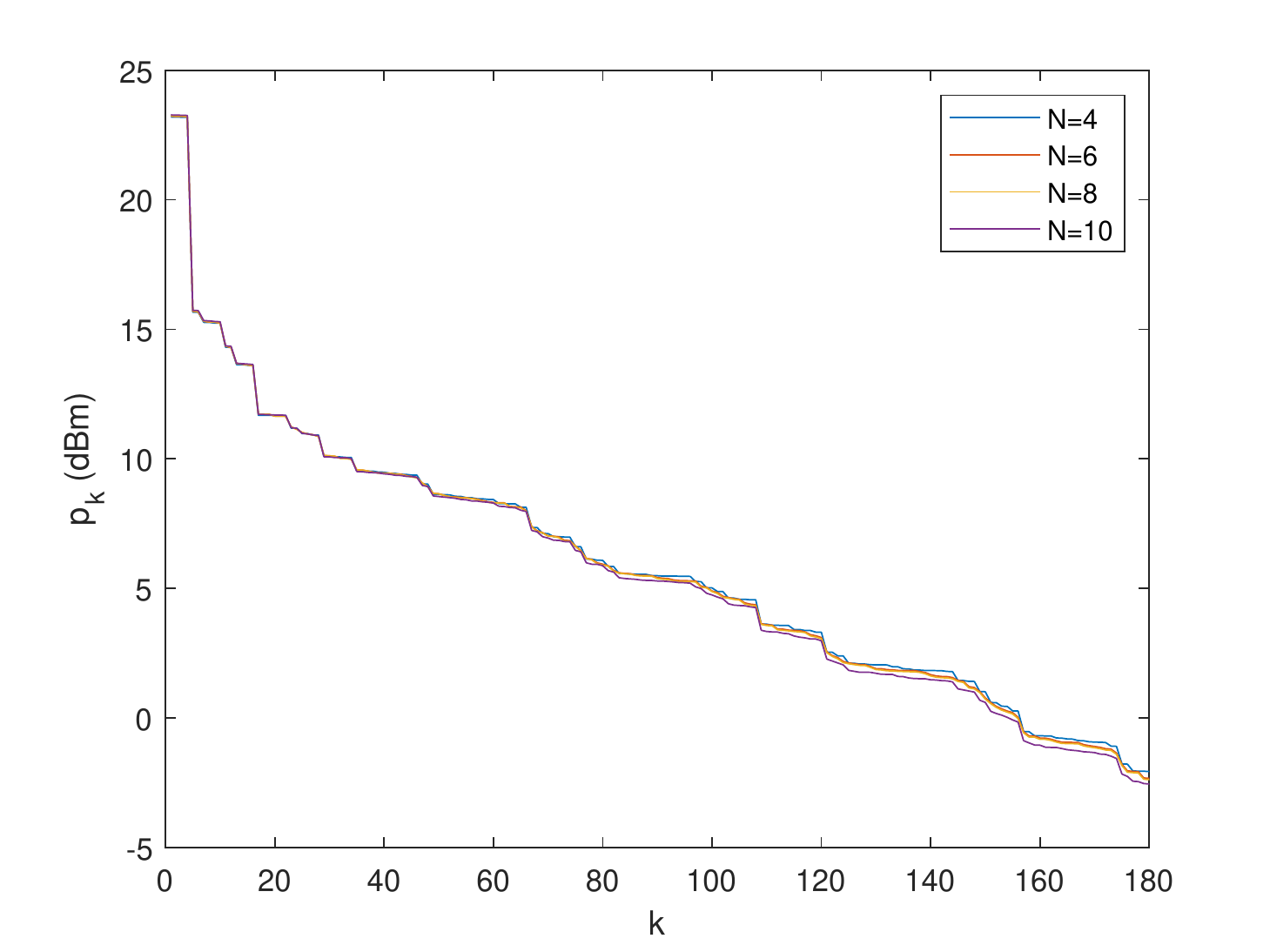}}
  \centerline{(a)}
\end{minipage}
\hfill
\begin{minipage}{0.49\linewidth}
  \centerline{\includegraphics[width=8cm]{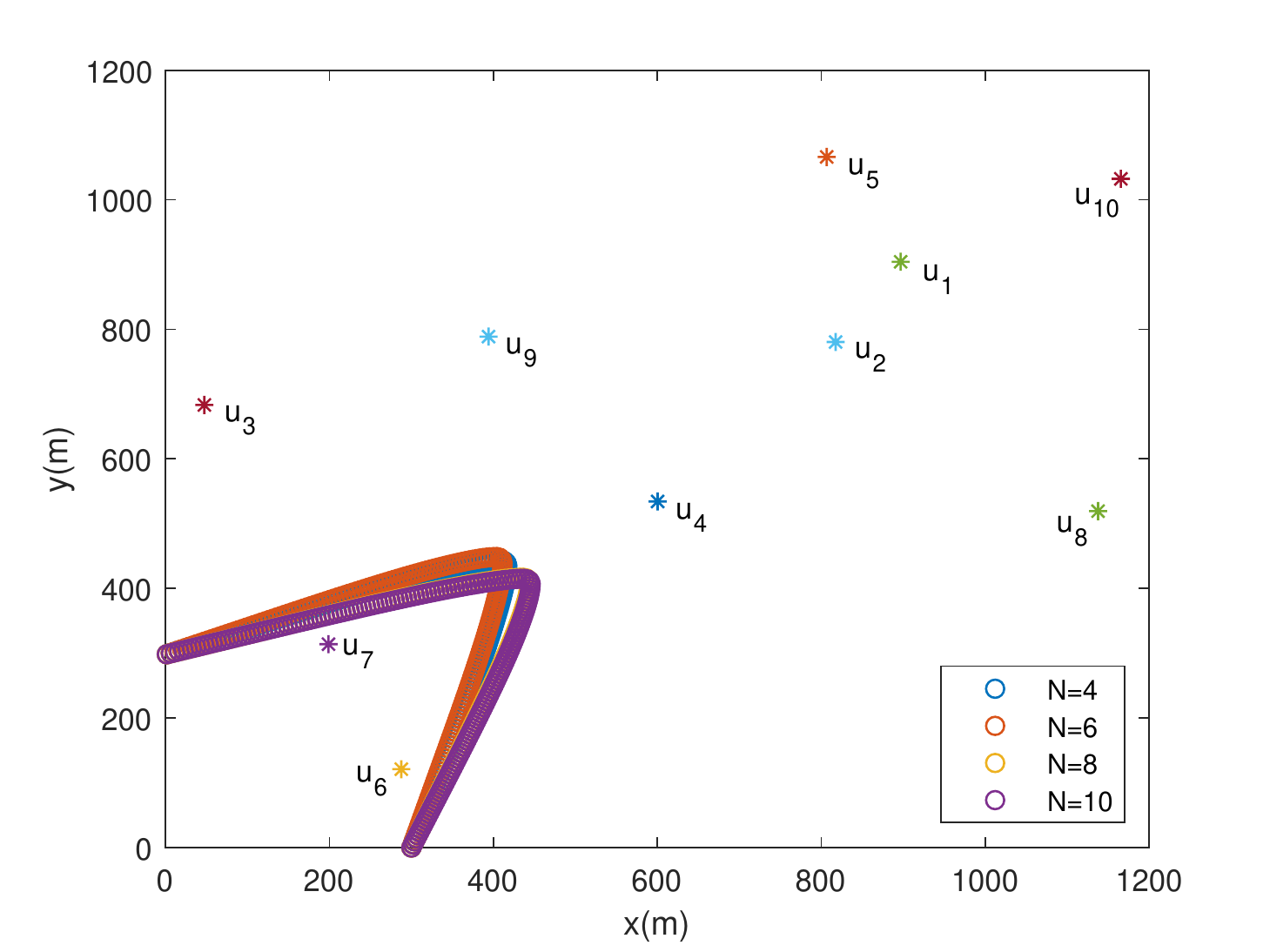}}
  \centerline{(b)}
\end{minipage}
\vfill
\begin{minipage}{0.49\linewidth}
  \centerline{\includegraphics[width=8cm]{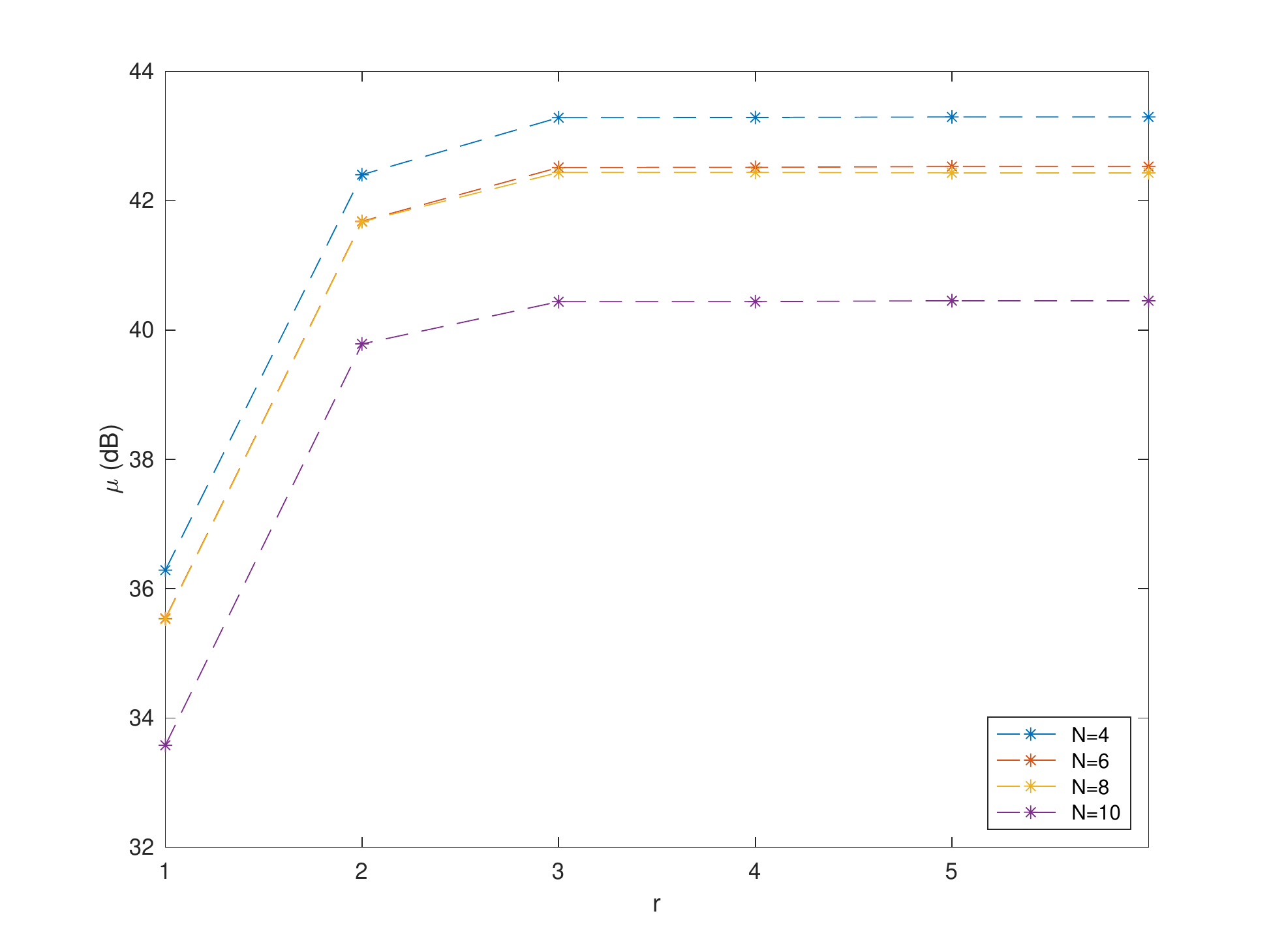}}
  \centerline{(c)}
\end{minipage}
\hfill
\begin{minipage}{0.49\linewidth}
  \centerline{\includegraphics[width=8cm]{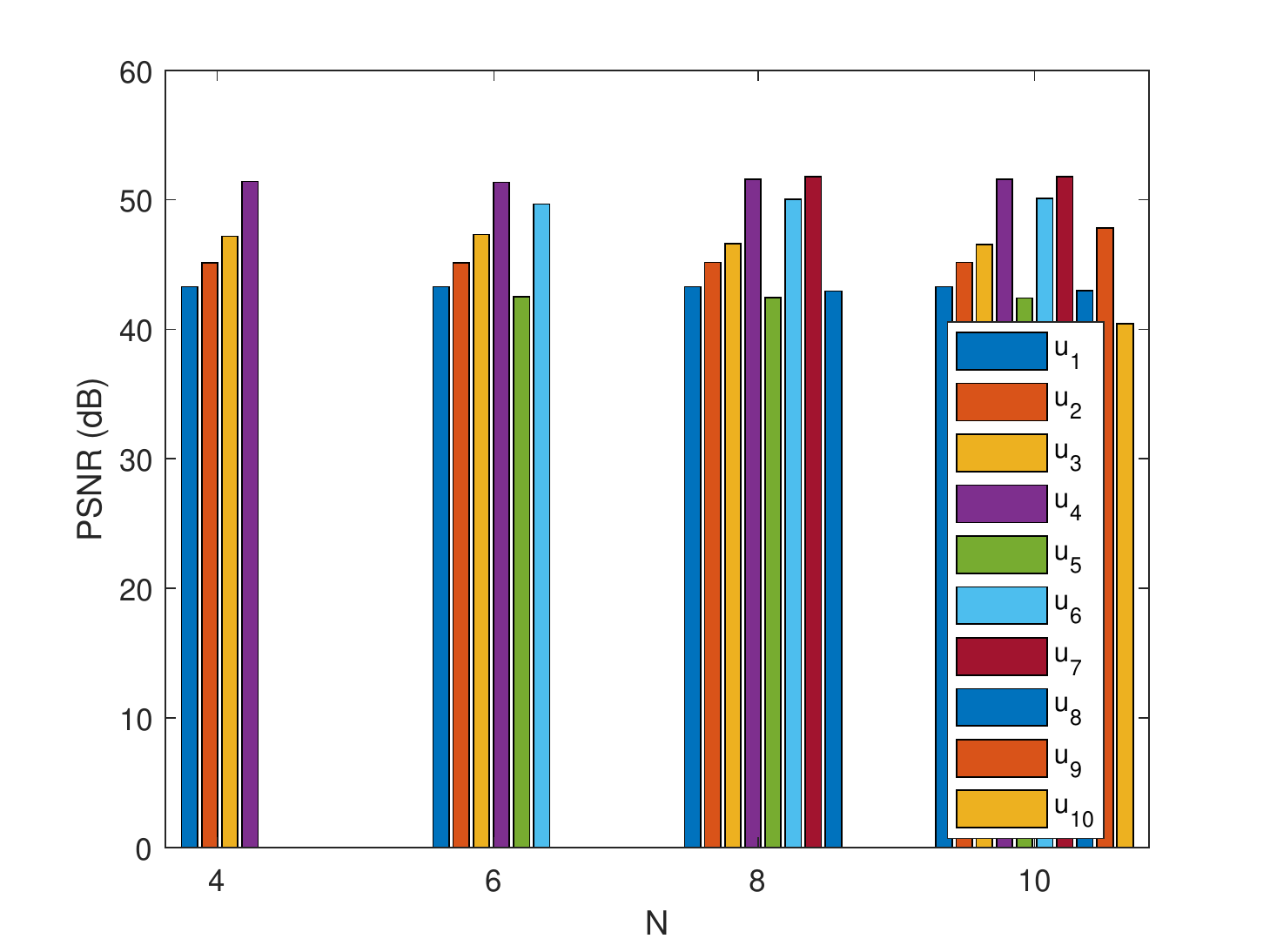}}
  \centerline{(d)}
\end{minipage}
\caption{Performance comparisons under conditions of different numbers of GUs: (a) transmission power allocation strategy, (b) UAV trajectory, (c) convergence of the proposed algorithm, (d) PSNRs of GUs.}
\label{fig:res}
\end{figure*}

Fig. 4(a) shows the transmission power allocation strategies under different values of $N$. We can see that the increase of $N$ does not have a big impact on the transmission power allocation strategy. From Fig. 4(b), one can see that due to the increase of $N$ and the dispersion of users' distribution, the UAV's trajectory gradually deviates from the users with the farthest distance. The UAV manages to fly as close as possible to the farthest GU to maximize its video demodulation quality.

Fig. 4(c) shows the convergence of proposed algorithm under different values of $N$. One can find that the converged value has a close relationship with the distribution of GUs. Specifically, the proposed system has similar performance under the cases of $N=6$ and $N=8$ since $u_5$ and $u_8$ have almost the same distance from the BS which suffers from the worst demodulation quality in each case. Fig. 4(d) shows the video demodulation qualities of all GUs under different values of $N$. With the increase of $N$, the GUs' minimum video demodulation quality gradually decreases. This is because with the increase of $N$, the distance between the cell-edge GU and the starting point also gets longer. However, the distance that a UAV can reach is limited by $E_t$ and $K$. Therefore, the video demodulation quality of the edge GU gradually declines as $N$ increases.

\subsection{Performance comparisons with different systems}
In this section, we compare the performance of the proposed system with three other classical video transmission systems, e.g., DVB, SoftCast [7], and SharpCast [8]. In DVB, we generate MPEG4 streams using the H.264/AVC codec provided by the FFmpeg software and the X264 codec library [54]. To ensure that all of the systems occupy the same channel bandwidth, the MPEG4 streams are encoded into the bit streams at ${1}/{3}$ code rate and mapped into complex signals using 16QAM. According to [8], we decompose the video into a content part and a structure part in SharpCast. The structure part is compressed by HEVC codec and is protected with a robust digital transmission scheme. The content part in SharpCast is transmitted in PAVT scheme. We assume that the BS is located at the origin of the coordinate axis in DVB, SoftCast, and SharpCast. For the sake of fairness, the four systems transmit the video signals with the same total energy. The performance of the four systems are compared from the perspectives of both subjective visual quality and objective evaluation metric.
\begin{figure}[htbp!]
\centering
\includegraphics[width=0.5\textwidth]{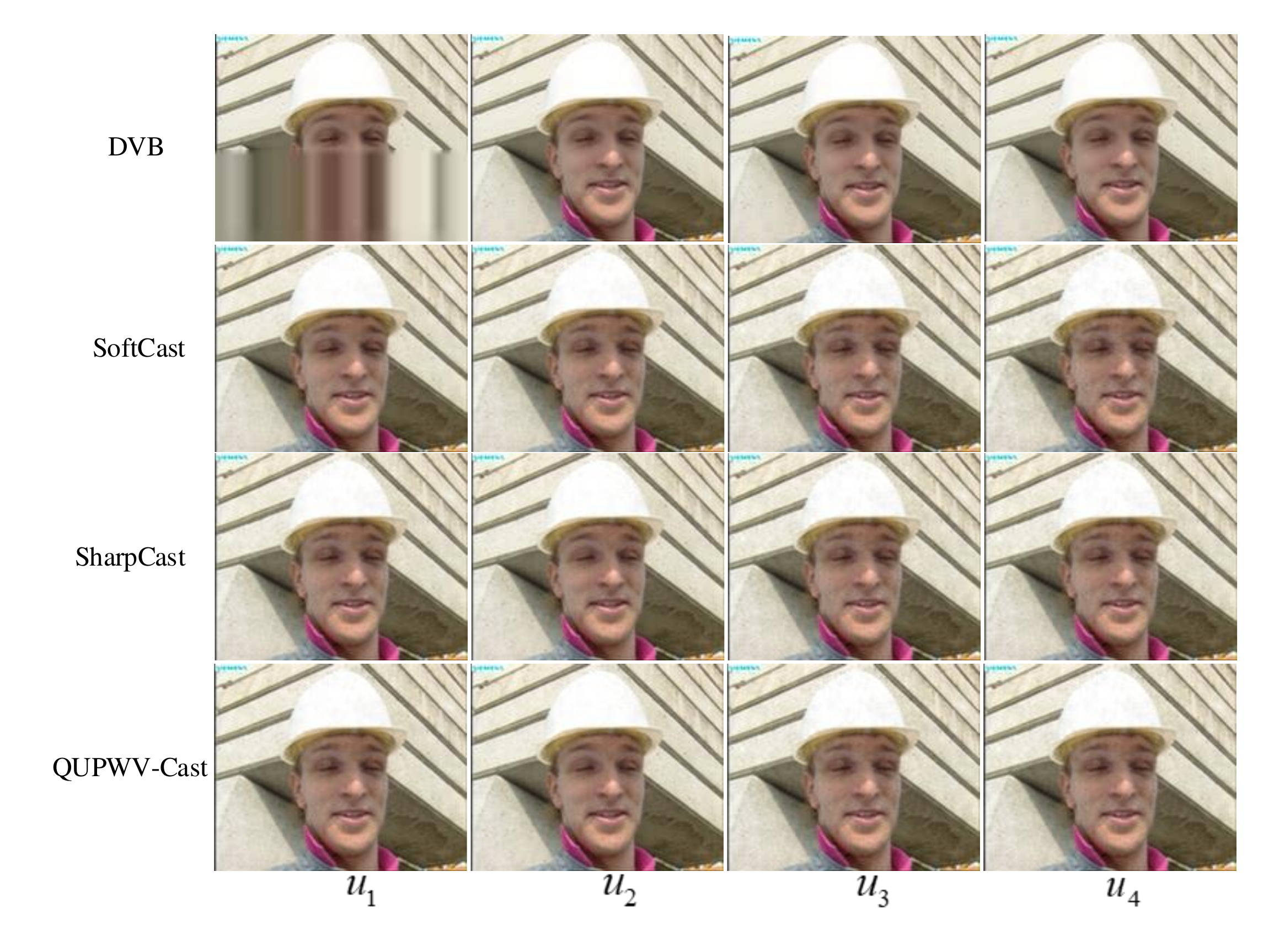}
\caption{Demodulated subjective visual quality of the $1^{st}$ frame using the four systems.}
\label{F1}
\end{figure}
\begin{figure}[htbp!]
\centering
\includegraphics[width=0.5\textwidth]{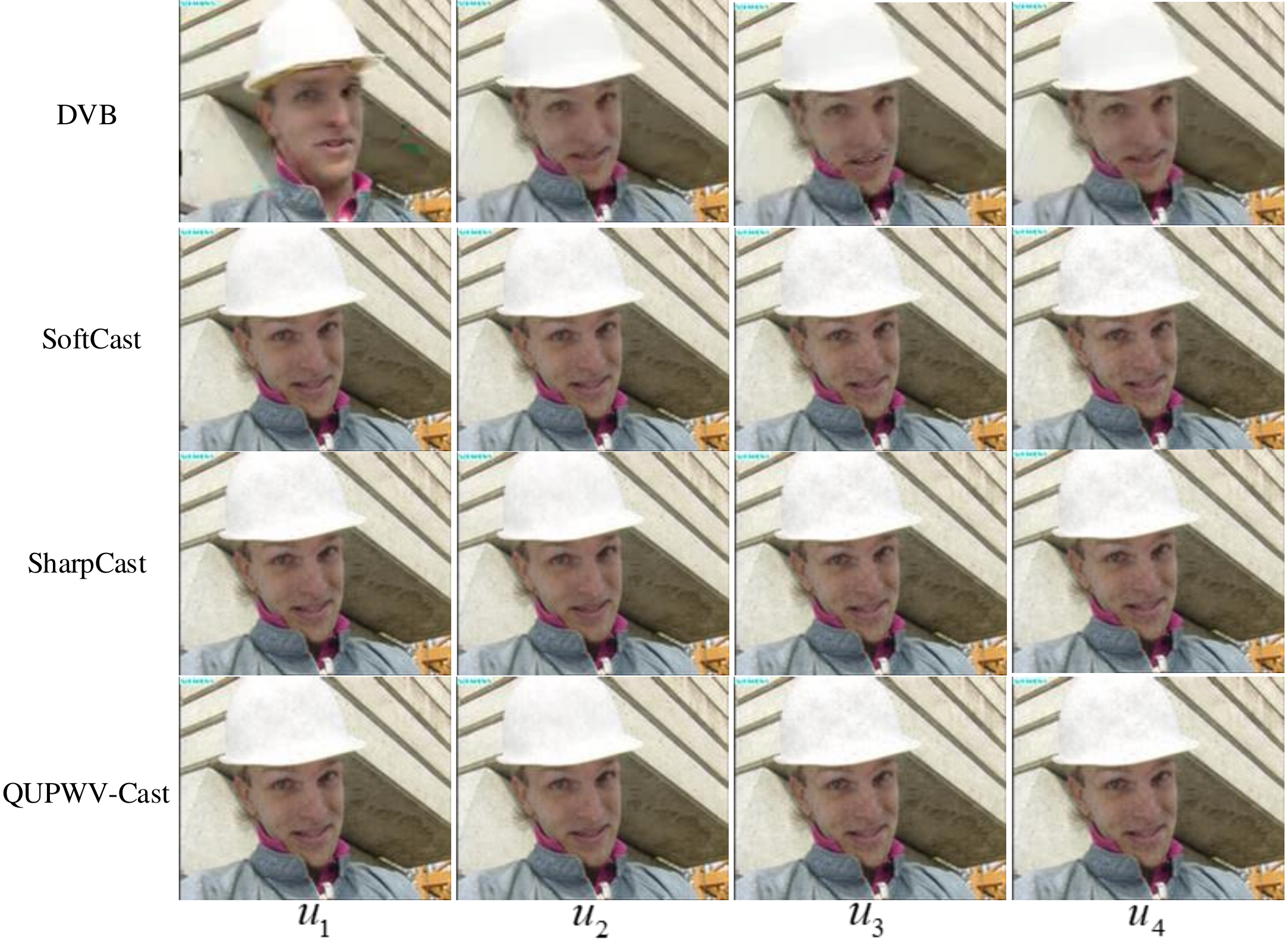}
\caption{Demodulated subjective visual quality of the $150^{th}$ frame using the four systems.}
\label{F1}
\end{figure}
\begin{figure}[htbp!]
\centering
\includegraphics[width=0.5\textwidth]{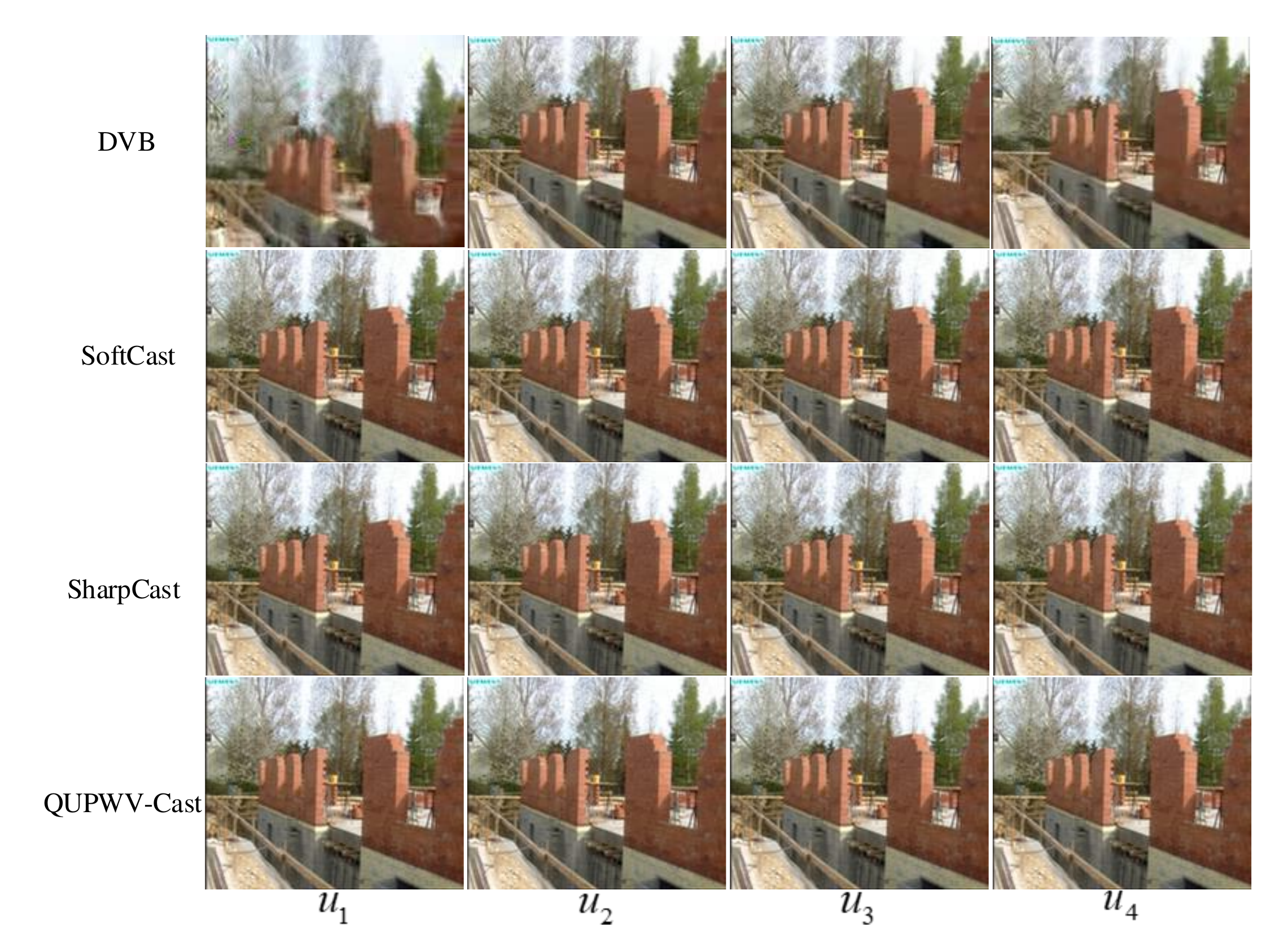}
\caption{Demodulated subjective visual quality of the $300^{th}$ frame using the four systems.}
\label{F1}
\end{figure}

Figs. 5-7 show the demodulated frames of four GUs using the four systems. One can conclude that $u_2$, $u_3$ and $u_4$ can always obtain excellent subjective visual qualities in all the four systems. $u_1$ gets high visual qualities in both SoftCast, SharpCast, and our proposed system. However, DVB can not provide $u_1$ with satisfactory subjective visual qualities compared with three other systems. Specifically, one can see from Figs. 5-7 that frames disorder when $u_1$ demodulates the video using DVB system. This is because DVB adopts the inter-frame compression and motion compensation, resulting in a high correlation between frames\footnote{In order to achieve high compression rate, the video is often divided into I-frames and P-frames in DVB. The demodulation of P-frames often relies on the I-frames.}. The demodulation error of a single frame may result in the loss of several seconds of video clips.

The detailed individual PSNR values of each GU using different testing sequences are provided in \textbf{Table III}. The minimum PSNR value in each method is marked with blue. One can see that the video demodulation quality of each GU using DVB is lower than three other systems. Table III also clearly demonstrates the cliff effect in DVB: when the channel quality is greater than a certain threshold, GUs' video demodulation quality can no longer be improved (the locations of $u_3$ and $u_4$ are different, but their video demodulation qualities are the same). This is due to the inherent loss caused by the lossy compression technology adopted by DVB. Among the four systems, the proposed system can obtain the best performance in terms of the poorest GU's reconstructed video quality (see Eq. (13)). Therefore, the proposed system can overcome the defect of GUs' geographical location with the help of the UAV's mobility.
\begin{table}[htbp!]
\centering
\caption{Individual PSNR values of GUs with different systems.}
\scalebox{1}{
\begin{tabular}{|c|c|c|c|c|c|}
\hline
\multirow{2}{*}{Sequences}& \multirow{2}{*}{Systems} & \multicolumn{4}{c|}{Average PSNR (dB)} \\ \cline{3-6}
& & $u_1$ & $u_2$ & $u_3$ & $u_4$ \\
\hline
\multirow{3}{*}{Foreman} & DVB &{\color{blue}19.50}&33.36&36.81&36.81 \\ \cline{2-6}
&SoftCast&{\color{blue}39.56}&40.57&44.79&43.45 \\ \cline{2-6}
&SharpCast&{\color{blue}40.44}&41.38&45.26&43.99 \\ \cline{2-6}
&QUPWV-Cast&{\color{blue}43.29}&45.15&46.91&51.51 \\ \cline{1-6}
\hline
\multirow{3}{*}{Akiyo} & DVB &{\color{blue}36.45} &40.09 &45.22 &45.22 \\ \cline{2-6}
&SoftCast &{\color{blue}40.65} &41.66 &45.84 &44.51 \\ \cline{2-6}
&SharpCast&{\color{blue}42.04}&43.05&47.22&45.21 \\ \cline{2-6}
&QUPWV-Cast &{\color{blue}44.52} &46.36 &48.11 &52.43 \\ \cline{1-6}
\hline
\multirow{3}{*}{Coastguard} & DVB &{\color{blue}18.76}&33.94 &34.93 &34.93 \\ \cline{2-6}
&SoftCast &{\color{blue}40.28}&41.27 &45.36 &44.07 \\ \cline{2-6}
&SharpCast&{\color{blue}40.99}&41.88&45.66&44.52 \\ \cline{2-6}
&QUPWV-Cast&{\color{blue}44.07} &45.82 &47.66 &51.27 \\ \cline{1-6}
\hline
\multirow{3}{*}{Container} & DVB &{\color{blue}18.71} &27.04 &42.40 &42.40 \\ \cline{2-6}
&SoftCast&{\color{blue}38.38} &39.38 &43.47 &42.18 \\ \cline{2-6}
&SharpCast&{\color{blue}39.37}&40.28&44.24&42.87 \\ \cline{2-6}
&QUPWV-Cast&{\color{blue}42.11} &43.83 &45.77 &49.17 \\ \cline{1-6}
\hline
\multirow{3}{*}{Hall} & DVB &{\color{blue}21.63} &36.07 &41.30 &41.30 \\ \cline{2-6}
&SoftCast&{\color{blue}38.82} &39.81&43.89&42.61 \\ \cline{2-6}
&SharpCast&{\color{blue}39.99}&41.28&44.66&43.88 \\ \cline{2-6}
&QUPWV-Cast&{\color{blue}42.72} &44.45 &46.30 &49.80 \\ \cline{1-6}
\hline
\multirow{3}{*}{Mother-daughter} & DVB &{\color{blue}23.82} &32.02 &43.47 &43.47 \\ \cline{2-6}
&SoftCast &{\color{blue}43.29} &44.29 &48.45&47.13\\ \cline{2-6}
&SharpCast&{\color{blue}45.31}&46.16&49.69&48.58 \\ \cline{2-6}
&QUPWV-Cast &{\color{blue}46.99} &48.80 &50.60 &54.72 \\ \cline{1-6}
\hline
\end{tabular}}
\label{T3}
\end{table}

In order to study the effect of the number of GUs on the system performance, we observe the average PSNR values of GUs in each method under the four cases: $N=\{4,6,8,10\}$ (please note that the GUs has the same coordinates with those in Part D). From Fig. 8, one can see that the proposed system can always achieve the best average PSNR performance in different cases. SharpCast can achieve a performance gain in terms of GUs' average PSNR compared with SoftCast. This is because SharpCast can preserve more structure-related information. DVB performs the worst in terms of the average PSNR among the four systems. Therefore, one can conclude that the PAVT system is more suitable for broadcast scenarios than DVB system. In addition, we can see from Fig. 8 that with the increase of the number of GUs, the average PSNR of GUs in each scheme has not changed much. This is because that we aim at minimizing the minimum PSNR of GUs rather than average PSNR. However, according to our analysis in Part D of Section V, we can conclude that the distribution of GUs rather than the number directly affects the performance of the proposed system.
\begin{figure}[htbp!]
\centering
\includegraphics[width=0.5\textwidth]{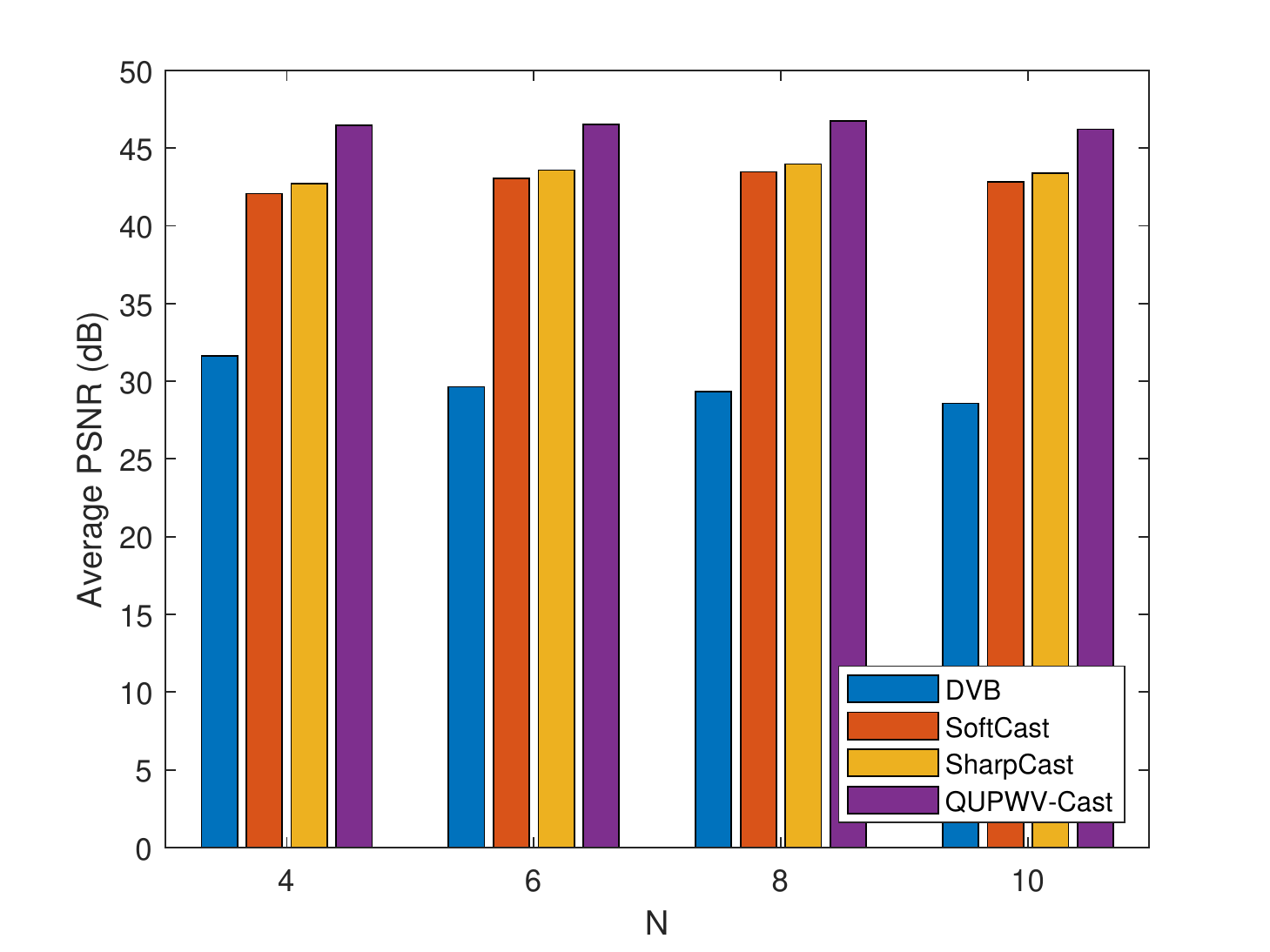}
\caption{The average PSNR under each case with different number of GUs.}
\label{F1}
\end{figure}

\section{Conclusions}
In this paper, a novel QoE-driven UAV-enabled pseudo-analog wireless video broadcast system called QUPWV-Cast, has been proposed to enhance the QoE of cell-edge GUs. The proposed system was modelled as a challenging non-convex optimization problem, aiming at maximizing the minimum PSNR of GUs by jointly optimizing the transmission power allocation strategy and the UAV trajectory. An efficient BCD and SCA-based algorithm was proposed to divide the original optimization problem into two sub-optimal problems. A sub-optimal solution could be obtained by iteratively solving the two sub-optimal problems. Comprehensive simulation results have been provided to prove the effectiveness of the proposed system. The results have shown that the proposed QUPWV-Cast system has the best performance, compared with three other systems, e.g., DVB, SoftCast, and SharpCast, in terms of both subjective visual quality and objective evaluation metric.

\section{Acknowledgement}
The authors would like to thank all reviewers for their efforts in reviewing this manuscript.

\begin{appendices}
\section{}
For ease of proof, we first rewrite the constraint (22) as follows
\begin{equation}\label{E41}
\begin{array}{l}
\frac{{M{\eta ^2}}}{{{\sum\limits_{k = 1}^K {\frac{{{\sigma_0^2}{\lambda _k}{{\left\| {{\bf{q}}[k] - {{\bf{w}}_n}} \right\|}^2}}}{{{\beta _0}{p_k}}}}  + \sum\limits_{m = K + 1}^M {{\lambda _m}} }}} \ge {10^{\frac{\mu }{{10}}}},\forall n.
\end{array}
\end{equation}

In Eq. (41), it is obvious that the right-hand side is convex. Therefore, we only need to prove that the left-hand side of Eq. (41) is concave. For simplicity, we make the following definition
\begin{equation}\label{E42}
\varphi_n ({p_k}) \buildrel \Delta \over = \frac{{{\gamma _0}}}{{\left( {\sum\limits_{k = 1}^K {\frac{{{\omega _n}[k]}}{{{p_k}}}} } \right) + {\gamma _1}}},\forall k, n,
\end{equation}
where ${\gamma _0} \buildrel \Delta \over = {{M{\eta ^2}}}$, ${\gamma _1} \buildrel \Delta \over = \sum\limits_{m = K + 1}^M {{\lambda _m}}$, and  ${\omega_n}[k] \buildrel \Delta \over = \frac{{{\sigma ^2}{\lambda _k}{{\left\| {{\bf{q}}[k] - {{\bf{w}}_n}} \right\|}^2}}}{{{\beta _0}}},\forall k,n$, respectively. The first-order and second-order partial derivatives of $\varphi_n$ with respect to ${p_k}$ can be denoted as follows
\begin{equation}\label{E43}
\frac{{\partial \varphi_n }}{{\partial {p_k}}} = \frac{{{\gamma _0}{\omega_n}[k]}}{{p_k^2{{\left( {\left( {\sum\limits_{k = 1}^K {\frac{{{\omega_n}[k]}}{{{p_k}}}} }\right) + {\gamma_1}} \right)}^2}}},\forall k,
\end{equation}
\begin{equation}\label{E42}
\begin{array}{l}
\frac{{{\partial ^2}\varphi_n }}{{\partial p_k^2}}\!=\!\frac{{-2{\gamma _0}{\omega _n}[k]\!\left(\!{\sum\limits_{k = 1}^K\!{\frac{{{\omega_n}[k]}}{{{p_k}}}}\!+\!{\gamma _1}}\!\right)\!\left(\!{{p_k}\!\left(\!{\sum\limits_{k = 1}^K\!{\frac{{{\omega _n}[k]}}{{{p_k}}}}\!+\!{\gamma _1}} \right)\!-\!{\omega _n}[k]}\!\right)\!}}{{p_k^4\left( {\left(\sum\limits_{k = 1}^K {\frac{{{\omega _n}[k]}}{{{p_k}}}}\right)  + {\gamma _1}} \right)}^4},\forall k,n.
\end{array}
\end{equation}
In Eq. (44), the inequality ${p_k}\left( {\left( {\sum\limits_{k = 1}^K {\frac{{{\omega_k[n]}}}{{{p_k}}}} } \right) + {\gamma _1}} \right) - {\omega _k[n]} \ge {p_k}\left( {\frac{{{\omega_k[n]}}}{{{p_k}}} + {\gamma _1}} \right) - {\omega_k[n]} = {\gamma _1} \ge 0$ holds. In addition, we have ${-2{\gamma _0}{\omega _n}[k]\!\left(\!{\sum\limits_{k = 1}^K\!{\frac{{{\omega_n}[k]}}{{{p_k}}}}\!+\!{\gamma _1}}\!\right)}<0$. Therefore, it is true that $\frac{{{\partial ^2}\varphi }}{{\partial p_k^2}} < 0$. Consequently, $\varphi_n$ is a concave function with respect to ${p_k}$.
\vspace{-2mm}
\section{}
We can prove the correctness of Lemma 2 by contradiction. If Lemma 2 doesn't hold, then we can always fix other variables and only increase $o[k]$ to make $||{\bf{v}}[k]|| = o[k],\forall k$ hold without changing the optimal solution to {\emph{P}}5.
\end{appendices}

\begin{IEEEbiography}[{\includegraphics[width=1in,height=1.25in,clip]{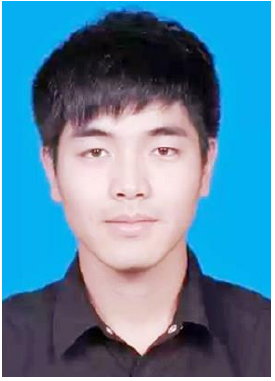}}]{Xiao-Wei Tang} (\emph{S'16, IEEE}) received the B.E. degree in Communication Engineering from Tongji University in 2016, where he is currently pursuing the Ph.D. degree. He has published several research papers on IEEE Transactions on Multimedia, IEEE Transactions on Vehicular Technology, IEEE Access, IEEE Globecom, and Mobile Networks \& Applications. He was a recipient of the Excellent Bachelor Thesis of Tongji University in 2016, the National Scholarship for Graduate Students by Ministry of Education of China in 2017, the Outstanding Students Award of Tongji University in 2017, the Outstanding Freshman Scholarship of Tongji University in 2018, the Chinese Government Scholarship by China Scholarship Council in 2019, the Outstanding Students Award of Tongji University in 2019, and the National Scholarship for Graduate Students by Ministry of Education of China in 2019. From Aug. 2019, he is doing research on UAV-enabled wireless video transmission in the Department of Electrical and Computer Engineering, the National University of Singapore, as a visiting scholar. His research interests include pseudo-analog video transmission, UAV communication, convex optimization, and deep learning.
\end{IEEEbiography}

\begin{IEEEbiography}[{\includegraphics[width=1in,height=1.25in,clip]{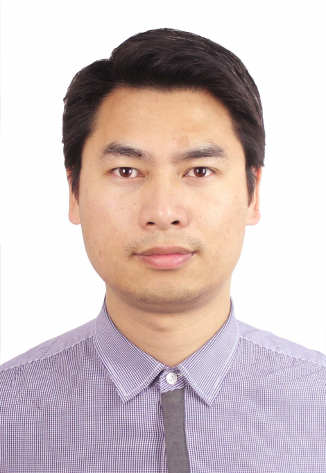}}]{Xin-Lin Huang} (\emph{S'09-M'12-SM'16, IEEE}) is currently a professor and vice-head of the Department of Information and Communication Engineering, Tongji University, Shanghai, China. He received the M.E. and Ph.D. degrees in information and communication engineering from Harbin Institute of Technology (HIT) in 2008 and 2011, respectively. His research focuses on Cognitive Radio Networks, Multimedia Transmission, and Machine Learning. He published over 70 research papers and 8 patents in these fields. Dr. Huang was a recipient of Scholarship Award for Excellent Doctoral Student granted by Ministry of Education of China in 2010, Best PhD Dissertation Award from HIT in 2013, Shanghai High-level Overseas Talent Program in 2013, and Shanghai Rising-Star Program for Distinguished Young Scientists in 2019. From Aug. 2010 to Sept. 2011, he was supported by China Scholarship Council to do research in the Department of Electrical and Computer Engineering, University of Alabama (USA), as a visiting scholar. He was invited to serve as Session Chair for the IEEE ICC2014. He served as a Guest Editor for IEEE Wireless Communications and Chief Guest Editor for International Journal of MONET and WCMC. He serves as IG cochair for IEEE ComSoc MMTC, and Associate Editor for IEEE Access. He is a Fellow of the EAI.
\end{IEEEbiography}

\begin{IEEEbiography}[{\includegraphics[width=1in,height=1.25in,clip]{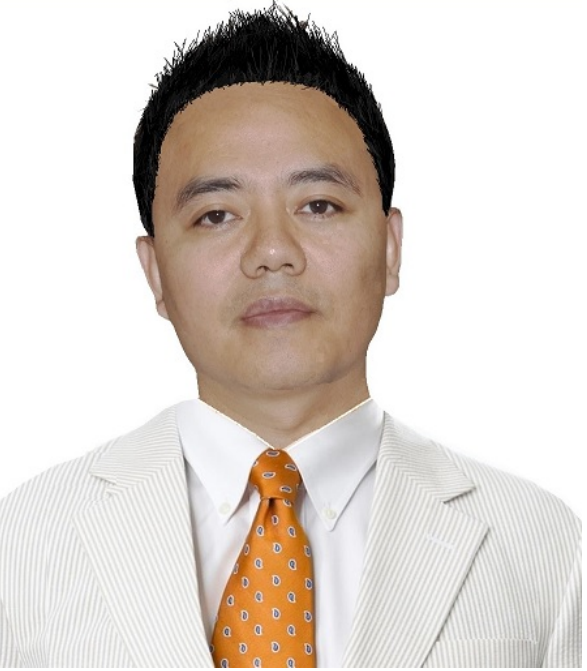}}]{Fei Hu}
(\emph{M'02, IEEE}) received the Ph.D. degree in signal processing from Tongji University, Shanghai, China, in 1999, and the Ph.D. degree in electrical and computer engineering from Clarkson University, Potsdam, NY, USA, in 2002. He is a Professor with the Department of Electrical and Computer Engineering, at the University of Alabama, Tuscaloosa, AL, USA. He has authored more than 200 journal/conference papers and book chapters in wireless networks, security, and machine learning. His research interests include cognitive radio networks, UAV communication, and cyber security.
\end{IEEEbiography}

\vfill

\end{document}